\documentclass[10pt]{article}
\usepackage{fullpage}
\usepackage[hidelinks]{hyperref}
\usepackage{booktabs}
\usepackage{amsfonts}
\usepackage{amssymb}
\usepackage{tikz}
\usepackage{amstext}
\usepackage{amsmath}
\usepackage{xspace}
\usepackage{theorem}
\usepackage{graphicx}
\usepackage{mdframed}
\usepackage{comment}
\usepackage{mathdots}
\usepackage{xcolor}
\usepackage{graphics}
\usepackage{colordvi}
\usepackage{color}
\usepackage{hyperref}
\usepackage{paralist}
\usepackage{bm}
\usepackage{mathpazo}
\usepackage{cleveref}
\usepackage{float}
\usepackage{kbordermatrix}
\usepackage[ruled,vlined]{algorithm2e}

\newtheorem{theorem}{Theorem}
\newtheorem{conjecture}{Conjecture}

\newtheorem{lemma}{Lemma}
\newtheorem{observation}{Observation}

\newtheorem{corollary}{Corollary}
\newtheorem{claim}{Claim}
\newtheorem{proposition}[theorem]{Proposition}

\newtheorem{definition}{Definition}

\newtheorem{remark}{Remark}

\newcommand{\lleft}{\mathsf{left}}
\newcommand{\rright}{\mathsf{right}}

\newcommand{\lset}{{\mathcal L}}
\newcommand{\rset}{{\mathcal R}}

\newcommand{\pset}{{\mathcal P}}

\newcommand{\nidia}[1]{\textcolor{blue}{#1}}
\newcommand{\opt}{{\sf OPT}\xspace}

\newenvironment{proof}{\par \noindent{\bf Proof:}}{\hfill\stopproof}
\def\stopproof{\square}
\def\square{\vbox{\hrule height.2pt\hbox{\vrule width.2pt height5pt \kern5pt
\vrule width.2pt} \hrule height.2pt}}

        {\hspace*{\fill}$\Box$\par\vspace{4mm}}

\newcommand{\gex}{{\sf gex}}

\usepackage{authblk}
\title{Improved Pattern-Avoidance Bounds for Greedy BSTs via Matrix Decomposition}

\author[1]{Parinya Chalermsook}
\author[2]{Manoj Gupta}
\author[1]{Wanchote Jiamjitrak}
\author[1]{Nidia Obscura Acosta}
\author[2]{Akash Pareek}
\author[1]{Sorrachai Yingchareonthawornchai}
\affil[1]{Aalto University,  \{parinya.chalermsook, wanchote.jiamjitrak\\nidia.obscuraacosta, sorrachai.yingchareonthwornchai\}@aalto.fi}
\affil[2]{IIT Gandhinagar, \{gmanoj,pareek\_akash\}@iitgn.ac.in}

\date{}

\begin{document}
\maketitle

\begin{abstract}
Greedy BST (or simply Greedy) is an online self-adjusting binary search tree defined in the geometric view ([Lucas, 1988; Munro, 2000; Demaine, Harmon, Iacono, Kane, Patrascu, SODA 2009). Along with Splay trees (Sleator, Tarjan 1985), Greedy is considered the most promising candidate for being dynamically optimal, i.e., starting with any initial tree, their access costs on any sequence is conjectured to be within $O(1)$ factor of the offline optimal. 
However, despite having received a lot of attention in the past four decades, the question has remained elusive even for highly restricted input.  

In this paper, we prove new bounds on the cost of Greedy  in the ``pattern avoidance'' regime. Our new results include:
\begin{itemize}
 
 \item The (preorder) traversal conjecture for Greedy  holds up to a factor of $O(2^{\alpha(n)})$, improving upon the bound of $2^{\alpha(n)^{O(1)}}$ in (Chalermsook et al., FOCS 2015) where $\alpha(n)$ is the inverse Ackermann function of $n$. This is the best known bound obtained by any online BSTs. 
  
 \item We settle the postorder traversal conjecture  for Greedy. Previously this was shown for Splay trees only in certain special cases (Levy and Tarjan, WADS 2019). 
 
 \item The deque conjecture for Greedy  holds up to a factor of $O(\alpha(n))$, improving upon the bound $2^{O(\alpha(n))}$ in (Chalermsook, et al., WADS 2015). This is arguably ``one step away'' from the bound $O(\alpha^*(n))$ for Splay trees (Pettie, SODA 2010).
  
  \item The split conjecture holds for Greedy up to a factor of $O(2^{\alpha(n)})$. Previously the factor of $O(\alpha(n))$ was shown for Splay trees only in a special case (Lucas, 1988).
\end{itemize}
The input sequences in traversal and deque conjectures are perhaps  ``easiest'' in the pattern-avoiding input classes and yet among the most notorious  special cases of the dynamic optimality conjecture. 
Key to all these results is to partition (based on the input structures) the execution log of Greedy into several simpler-to-analyze subsets for which classical forbidden submatrix bounds can be leveraged. We believe that this simple method will find further applications in doing amortized analysis of data structures via extremal combinatorics. 
Finally, we show the applicability of this technique to handle a class of increasingly complex pattern-avoiding input sequences, called {\em $k$-increasing sequences}.  

As a bonus, we  discover a new class of permutation matrices whose extremal bounds are polynomially bounded. This gives a partial progress on an open question by Jacob Fox (2013). 

\end{abstract}

\clearpage 

\setcounter{page}{1} 

\section{Introduction}

The dynamic optimality conjecture  postulates that there exists an online binary search tree (BST) whose cost to serve any input sequence (search, insert, delete) is at most the optimal offline cost of any binary search tree. The two most promising candidates for being dynamically optimal are Splay trees ~\cite{sleator1985self} and Greedy~\cite{demaine2009geometry,munro2000competitiveness,lucas1988canonical}.
Despite continuing efforts for many decades (see, e.g.,the surveys and monographs~\cite{iacono2013pursuit,kozma2016binary,chalermsook2016landscape}), the conjecture remains wide open even for highly restricted corollaries of the conjecture.   We describe some of the most important conjectures that fall under the regime of pattern avoidance: Splay trees and Greedy satisfy preorder traversal, postoder traversal, deque and split properties where these properties are defined below. 

\vspace{0.1in} 

\begin{minipage}{0.9\textwidth} 
\begin{mdframed}
{\bf (Preorder) Traversal Property: } An online BST satisfies \textit{preorder traversal property} if, starting with any initial BST on $n$ keys, for any input $X = (x_1,\ldots, x_n) \in [n]^n$ obtained by \textit{preorder} traversal of (potentially distinct) binary search tree $R$ on $[n]$, it searches $X$ with $O(n)$ cost.\footnote{More formally, $x_t$ is searched at time $t$ for all $t \in [n]$.} 
\end{mdframed}
\end{minipage}

\vspace{0.1in} 

\begin{minipage}{0.9\textwidth} 
\begin{mdframed}
{\bf (Postorder) Traversal Property: } An online BST satisfies \textit{postorder traversal property} if, starting with any initial BST on $n$ keys, for any input $X = (x_1,\ldots, x_n) \in [n]^n$ obtained by \textit{postorder} traversal of (potentially distinct) binary search tree $R$ on $[n]$, it searches  $X$ with $O(n)$ cost.
\end{mdframed}
\end{minipage}

\vspace{0.1in}
\begin{minipage}{0.9\textwidth} 
\begin{mdframed}
{\bf Deque Property:} An online BST satisfies \textit{deque property} if, for $m>n$, starting with any initial BST on $n$ keys, can serve $m$  operations {\sc InsertMin}, {\sc InsertMax}, {\sc DeleteMin}, {\sc DeleteMax} in $O(m)$ time.
\end{mdframed}
\end{minipage}

\vspace{0.1in} 

\begin{minipage}{0.9\textwidth} 
\begin{mdframed}
{\bf Split Property:}  An online BST satisfies \textit{split property} if,  starting with any initial tree with $n$ keys, serves a sequence of $n$ {\sc Split} operations in time $O(n)$ where the operation {\sc Split}($i$) moves $i$ to the root and then deletes it, leaving two independent binary search (split) trees.
\end{mdframed}
\end{minipage}

\vspace{0.1in} 

This paper focuses on Greedy's bounds for such corollaries through the lens of ``pattern avoidance'' (to be made precise later). Each of them is of independent interest and therefore has received a lot of attention in the literature. Resolving these conjectures (especially the preorder conjecture) is considered the ``simplest'' step of dynamic optimality conjecture and yet has so far resisted attempts for past three decades. 

\paragraph{Bounds on Splay Trees.} For deque property,  Splay trees have been shown to cost at most $O(m \alpha(n))$ by Sundar~\cite{sundar1992deque} and later $O(m \alpha^*(n))$ by Pettie~\cite{pettie2007splay}. It has remained open whether Splay's cost is $o(n\log n)$ for preorder and postorder traversals.  Special cases when we start inserting preorder or postorder sequence $X$ from an empty-initial tree were resolved recently by Levy and Tarjan~\cite{levy2019splaying}.  Lucas~\cite{lucas1992competitiveness} showed that the split costs $O(n\alpha(n))$ in Splay trees when the initial tree is a path.  

\paragraph{Bounds on Greedy.} The bounds known for Greedy are generally better than the Splay's counterpart (except for deque). For deque property, Greedy is known to cost $m 2^{O(\alpha(m,m+n))}$~\cite{DBLP:journals/corr/ChalermsookG0MS15a}. For both preorder and postorder traversal sequences, Greedy is known to cost at most $n2^{\alpha(n)^{O(1)}}$
 ~\cite{chalermsook2015pattern}. We are not aware of published results for Split conjecture. Greedy algorithm is formally defined in \Cref{section: short-prelims}.

\paragraph{Remark.} One can ask popular conjectures of BST in two settings: (1) when the initial BST can be pre-processed or (2) when it cannot be pre-processed. There is a gap in our understanding of these two settings. For example, it is not known if Greedy's cost for preorder traversal is same in both the settings. In setting (1), Chalermsook et al. \cite{chalermsook2015pattern} showed that Greedy takes $O(n)$ for the preorder traversal. One can also solve this problem using the ideas in \cite{iacono2016weighted}. If preprocessing is allowed, then Splay trees cost $O(n)$ for the preorder traversal \cite{ChaudhuriH93}. We will consider setting (2) in this paper. We also note that Multi-Splay trees  satisfy deque property \cite{multisplay}.

\paragraph{Broader context: Amortized analysis and forbidden submatrix theory.} 
Resolving these conjectures represents a small part of a much broader algorithmic challenge in amortized analysis of online algorithms/data structures. 
Amortized analysis is typically done via potential function method, which would be easier when the algorithm designer is allowed to tailor an algorithm towards a tentative analytical method they have in mind. However, in the context of analyzing Greedy or Splay, the algorithms are already fixed in advance (e.g. these are the algorithms that tend to work well in practice), so we have no control on the ``design'' part. In such cases, the state-of-the-art understanding of potential function design is much more adhoc and mostly tailored to specific cases. Indeed, there has been no systematic, efficient way known for the task of designing a potential function, when an algorithm is fixed in advance. The fact that the aforementioned conjectures have remained open for decades clearly underlines the lack of understanding on this front.

Extremal combinatorics methods (such as forbidden submatrix theory) have been used successfully in amortized analysis as an alternative to potential function design. 
In the context of binary search trees, such attempts were pioneered by Pettie~\cite{pettie2007splay,pettie2010applications} and more recently extended by~\cite{DBLP:journals/corr/ChalermsookG0MS15a,chalermsook2015pattern}.  Informally speaking, one can encode an execution log of Greedy as a binary matrix. It is a non-trivial fact that, since the input is restricted, the execution log of Greedy is also restricted. When the execution log is restricted, it cannot have too many 1's in the matrix, and thus we can apply the extremal bounds from forbidden submatrix theory as a black box.

More precisely, forbidden submatrix theory is a collection of theorems of the form: Let $\pi$ be a matrix (pattern), and ${\sf ex}(n, \pi)$ denotes the extremal bound which equals the maximum number of $1$s in any $n$-by-$n$ $0/1$ matrix that \textit{avoids} pattern $\pi$ (a matrix $M$ \textit{contains pattern} $\pi$ if it is possible to obtain $\pi$ from $M$ by removing rows, columns, and turning ones into zeroes; otherwise, we say that $M$ avoids $\pi$), see \Cref{fig:pattern} for illustration. Studying behavior of extremal functions for various matrices $\pi$ have been a fruitful area of research in extremal combinatorics. 

\begin{figure}[ht]
  \centering
  \includegraphics[scale=0.27]{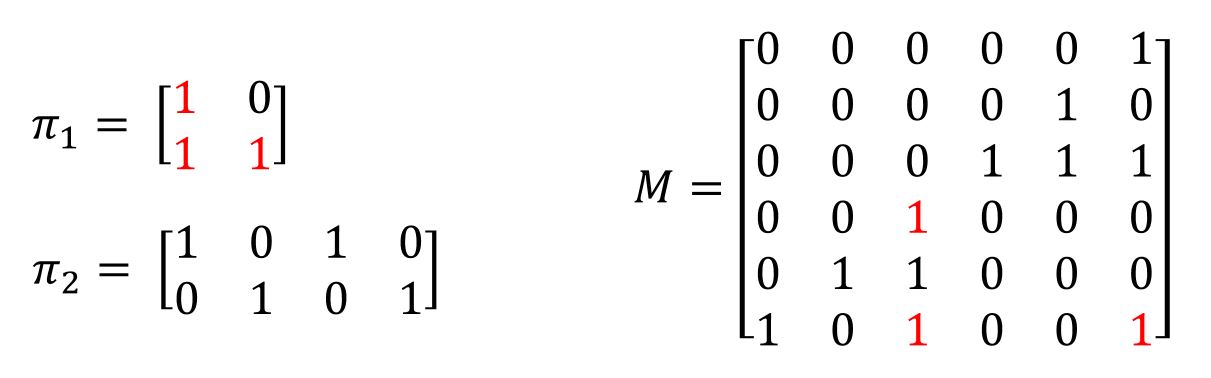}
  \caption{An example that $M$ contains pattern $\pi_1$, but avoids $\pi_2$.\label{fig:pattern}} 
\end{figure}

Let $X \in [n]^n$ be an input sequence. Denote by $G_X$ the matrix that ``encodes'' the execution log of Greedy, that is, $G_X(i,j) =1$ if and only if key $i$ is touched by Greedy at time $j$, implying that the number of $1$s in $G_X$ (denoted by $|G_X|$) is equal to the cost of the algorithm. The connection between BSTs and the theory of forbidden matrices (see, e.g., in~\cite{DBLP:journals/corr/ChalermsookG0MS15a,chalermsook2015pattern}) relies on a ``reduction statement'', which says that if $X$ avoids a pattern $\pi$ of size $k$, then the Greedy matrix $G_X$ avoids a (tensored) pattern $\pi'$ of size $3k$. Therefore, existing extremal bounds can be immediately used to upper bound $|G_X|$.  Indeed, for preorder and postorder traversals (that avoid patterns of size $3$), $G_X$ avoids a pattern of size $9$. The bound of $n 2^{\alpha(n)^{O(1)}}$ follows from this generic reduction. Here, we state the known reductions. 

\begin{lemma} [\cite{DBLP:journals/corr/ChalermsookG0MS15a,chalermsook2015pattern}] \label{lem:overall}
Let $Q_0 = \kbordermatrix{
     &  &  & & &  \\
      & 1 & &1 && 1 \\ 
    &   & 1& & 1 &}$, and let  $Q_1$ and $Q_2$ be the following matrices.
$$ Q_1 = \kbordermatrix{
     &  &  & & &  &  &  & &   \\
     &  & 1 & & &  &  &  & &   \\
     & 1 &  &1 & &  &  &  & &  \\
     &  &  & & &  &  &  &1 &   \\
     &  &  & & &  &  & 1 & & 1  \\
     &  &  & & &  1&  &  & &   \\
     &  &  & &1 &  & 1 &  & &  }, \mbox{ and } 
     Q_2 = \kbordermatrix{
     &  &  & & &  &  &  & &   \\
     &  &  & & &  1&  &  & &   \\
     &  &  & &1 &  & 1 &  & & \\
     &  &  & & &  &  &  &1 &   \\
     &  &  & & &  &  & 1 & & 1  \\
       &  & 1 & & &  &  &  & &   \\
     & 1 &  &1 & &  &  &  & &   
     }. $$
  (we omit zero entries for clarity).
 \begin{itemize} 
 \item If $X$ is delete-only deque sequence, then $G_X$ avoids    $Q_0$. Therefore, $|G_X| \leq {\sf ex}(n,Q_0) = O(n2^{\alpha(n)})$.
 \item If $X$ is preorder traversal, then $G_X$ avoids $Q_1$. Therefore, $|G_X| \leq {\sf ex}(n,Q_1) \leq n2^{\alpha(n)^{O(1)}}$. 
 \item If $X$ is postorder traversal, then $G_X$ avoids $Q_2$. Therefore, $|G_X| \leq {\sf ex}(n,Q_2) \leq n2^{\alpha(n)^{O(1)}}$. 
\end{itemize}
\end{lemma}


\paragraph{Barrier for Improvements.}
Given that \Cref{lem:overall} provides near optimal deque and traversal sequences for Greedy up to a factor of $2^{\alpha(n)^{O(1)}}$,  it is an intriguing open question to further extend this technique to settle deque and/or traversal conjectures for Greedy. As suggested by \cite{chalermsook2015pattern}, the improvement can potentially be made either finding a better pattern that is avoided by Greedy matrix, or improving the analysis of the extremal bounds for $Q_1$ and $Q_2$. However, there are some inherent limitations to this approach. First, the lower bounds are known\footnote{They contain the pattern $\kbordermatrix{
     &  &  & & \\
      &  & & &1  \\ 
      &  & 1&  &\\ 
    &  1 & & 1}$ whose extremal bound is $\Omega(n\alpha(n))$ \cite{FurediH92}.}  for $Q_1$ and $Q_2$. That is, ${\sf ex}(n,Q_1) = \Omega(n \alpha(n))$ and  ${\sf ex}(n,Q_2) = \Omega(n \alpha(n))$.  Although they are far from the current upper bounds, these results stop us from obtaining linear upper bounds. 
    
    The second key barrier is due to a counterexample  provided in~\cite{chalermsook2015pattern}. A natural way to prove that the Greedy matrix $G(X)$ satisfies $|G(X)| = O(n)$ is to show that $G(X)$ avoids a constant-sized permutation pattern, and applying the upper bound from~\cite{marcus2004excluded}.     However,~\cite{chalermsook2015pattern} shows a family of sequences $X$ such that $G(X)$ contains every constant-sized permutations even when the input is delete-only deque sequence. The counterexample also suggests that it is unlikely that the Greedy matrix will avoid some \textit{linear patterns} (i.e. pattern $\pi$ whose extremal bound ${\sf ex}(n,\pi) = O(n)$).

\paragraph{Our Results.} 
In this paper, in order to bypass the barriers, we propose to decompose the Greedy matrix $G(X)$ into several matrices that are ``easier in different ways''. More formally, we write $G(X) = \sum_{i = 1}^{\ell} M_i$ where matrices $M_i$ are chosen based on the structures of $X$ so that each $M_i$  avoids a much smaller pattern $\pi_i$ (which can be different for distinct $i$). This would give the upper bound $|G(X)| \leq \sum_{i=1}^{\ell} {\sf ex}(n, \pi_i)$. All our results follow this framework. We believe that our matrix decomposition techniques  will inspire further development of amortized analysis using extremal combinatorics beyond BSTs. 

\begin{itemize}
    
    \item For preorder traversal input $X$, we have 
    \begin{eqnarray*}
      |G(X)|& \leq & {\sf ex}\bigg(n,\kbordermatrix{
     &  &   \\
     & 1 &   \\ 
    &   & 1 } \bigg) + {\sf ex}\bigg(n,  \kbordermatrix{
     & & &  \\
     & 1 & &1  \\ 
    &   & 1& } \bigg) + 2\cdot {\sf ex}\bigg(n, \kbordermatrix{
     &  &  & & &  \\
      & 1 & &1 && 1 \\ 
    &   & 1& & 1 &} \bigg) 
    \end{eqnarray*} 
    which impies that $|G(X)| \leq O(n 2^{\alpha(n)})$ (details in \Cref{sec:size3input}). We remark that, without the matrix decomposition technique, the matrix $G(X)$ itself contains pattern $\kbordermatrix{
     &  &  & & &  \\
      & 1 & &1 && 1 \\ 
    &   & 1& & 1 &}$ (see \Cref{sec: counter examples} for a counterexample).  
    
    \item For postorder traversal input $X$, we have 
    \begin{eqnarray*}
      |G(X)|& \leq & {\sf ex}\bigg(n,\kbordermatrix{
     &  &   \\
     & 1 & 1  \\ 
    &   & 1 } \bigg) + {\sf ex}\bigg(n,  \kbordermatrix{
     & & &  \\
     & 1 & &1  \\ 
    &   & 1& } \bigg) +  2 \cdot {\sf ex}\bigg(n,  \kbordermatrix{
     & & &  \\
     &  & 1&  \\ 
    &  1 & & 1 } \bigg)
    \end{eqnarray*}
    which implies that $|G(X)| = O(n)$ (details in \Cref{sec:size3input}).
    
    \item For delete-only deque input $X$, we have 
    \[|G(X)| \leq {\sf ex}\bigg(n,\kbordermatrix{
     &  &   & &  \\
      &  & 1 &   &  1 \\ 
    &  1  & &  1  &  }\bigg) + {\sf ex}\bigg(n, \kbordermatrix{
     &  &   & &  \\
      & 1 &  & 1  &   \\ 
    &   &1 &    & 1 }\bigg) + O(n)\] 
    
    which implies that $|G(X)| \leq O(n \alpha(n))$ (details in \Cref{section: seq}).  We remark that, without the matrix decomposition technique, the matrix $G(X)$ itself  contains pattern $\kbordermatrix{
     &  &  & &   \\
      & 1 & &1 & \\ 
    &   & 1& & 1}$ and$\kbordermatrix{
     &  &  & &   \\
      &  &1 & &1 \\ 
    &   1& & 1& }$ (see \Cref{sec: counter examples} for a counterexample).  . 
\end{itemize}

We summarize our results in the following. 

\begin{theorem} \label{thm:main}
The following bounds hold for Greedy: 
\begin{enumerate}
    \item \label{item:main2}Greedy searches any preorder traversal sequence with cost $O(n 2^{\alpha(n)})$. 
    
    \item \label{item:main3}Greedy searches any postorder traversal sequence with cost $O(n)$. 
    
     \item \label{item:main1} Starting with any initial BST $R$ with $n$ keys, Greedy serves $m$ operations of {\sc InsertMin}, {\sc InsertMax}, {\sc DeleteMin}, {\sc DeleteMax} with cost at most $O(m \alpha(n))$ assuming $m>n$. 
    
\end{enumerate}
\end{theorem}

\begin{remark}
The sequence for preorder and postorder is a permutation sequence of length $n$. For deque we consider any sequence of length $m>n$.
\end{remark}

\begin{remark}
In all our results we consider an initial tree $T$ before the execution of Greedy. If the initial tree is not given, then problems becomes much "easier" to solve. For example the cost of sequential access is $O(n)$ with initial tree and without initial tree but the proof structures are different~\cite{chalermsook2015pattern,fox2011upper}.
\end{remark}

\begin{table}[]
\centering

\begin{tabular}{lccc}
\toprule
          & \textbf{Previous known}               & \textbf{This paper}                                 & \textbf{Remark} \\ \hline
Preorder       & $n2^{\alpha(n)^{O(1)}}$ \cite{chalermsook2015pattern} & $O(n2^{\alpha(n)})$                 & \Cref{thm:main}(\ref{item:main2}) \\  \hline
Postorder      & $n2^{\alpha(n)^{O(1)}}$ \cite{chalermsook2015pattern} & $O(n)$                              &\Cref{thm:main}(\ref{item:main3})  \\  \hline
Deque          & $O(m2^{\alpha(m,m+n)})$ \cite{DBLP:journals/corr/ChalermsookG0MS15a}  & $O(m\alpha(n))$                     &   \Cref{thm:main}(\ref{item:main1})\\  \hline
Split          & -                       & $O(n2^{\alpha(n)})$                 &  \Cref{thm:main split}\\  \hline
$k$-Increasing & $O(nk^2)$   \cite{chalermsook2015pattern,cibulka2013extremal}       & $O(\min\{nk^{2}, nk \alpha(n) \})$ &  \Cref{thm:main k-dec} \\ \bottomrule
\end{tabular}
\caption{Main results for Greedy BSTs. \label{table:results}}
\end{table}

For split conjecture, we in fact prove that the maximum cost of Greedy's splitting is at most the cost of searching a preorder traversal and therefore our traversal bound directly gives an upper bound on Greedy's serving split tree operations.   

\begin{theorem}[Informal]
\label{thm:main split} For $n \in {\mathbb N}$, let ${\sf deldeq}(n)$  be the maximum possible costs of Greedy when serving deletion-only deque operations on $n$ keys, and  ${\sf preorder}(n)$ be the maximum cost

 when serving preorder search. Then, 
\[{\sf deldeq}(n) \leq {\sf split}(n) \leq {\sf preorder}(n).\] 
Consequently, Greedy can be used as a split tree with cost $O(n 2^{\alpha(n)})$.
\end{theorem}

As a consequence for Greedy, the traversal property implies the split property, which implies the delete-only deque property. The implication from traversal to split properties is not known for Splay trees. 

One can view this collection of conjectures as the dynamic optimality conjecture on restricted inputs, where the restriction on the input sequence is defined by pattern avoidance properties.  Pattern avoiding problems are interesting special cases of the dynamic optimality conjecture that have shown  interplay between extremal combinatorics and data structures.

We now define pattern avoidance formally. Consider any input $X$.\footnote{It was argued in~\cite{demaine2007dynamic} that one can assume w.l.o.g. that the input is a permutation.} We say that $X=(x_1, x_2, \ldots, x_n)$ contains pattern $\pi = (\pi_1,\ldots, \pi_k)$ if there are indices $i_1 < \ldots < i_k$ such that the subsequence $(x_{i_1}, x_{i_2}, \ldots, x_{i_k})$  is order-isomorphic to $\pi$. 
Otherwise,  $X$ avoids $\pi$. 

The three properties can be (roughly) rephrased in this language as follows. For the preorder traversal property, we are given an input permutation $X$ that avoids pattern $(2,3,1)$, and our goal is to show a binary search tree that searches this sequence with cost at most $O(n)$. For the postorder traversal property, our goal is to search an input sequence that avoids $(1,3,2)$. For the deque property, if we represent the (delete-only) input $X$ (where $x_t$ is the key deleted at time $t$), then $X$ avoids patterns $(2,3,1)$ and $(2,1,3)$. 
In this way, all these properties deal with size-$3$ pattern-avoiding input classes.

Besides the small patterns, recent works have also started exploring the complexity of input classes that avoid patterns of growing sizes~\cite{chalermsook2015pattern,goyal2019better,chalermsook2016landscape}. 
See the thesis of Kozma~\cite{kozma2016binary} for more detail about this connection and a broader perspective on this class of problems. 
Our next result shows the improvement on a pattern-avoiding input class that allow patterns to be growing in terms of $k$. We say that an input $X$ is $k$-\textit{increasing} (respectively, $k$-\textit{decreasing}) if $X$ avoids $(k+1,k,k-1,\ldots,1)$ (respectively, $(1,2,\ldots,k,k+1)$). Note that $1$-increasing (1-decreasing) sequence corresponds to sequential sequences: $X = (1,\ldots,n)$ (or $X = (n,\ldots,1)$, respectively). The sequential sequence has been studied in the early days of the dynamic optimality conjecture~\cite{tarjan1985sequential,elmasry2004sequential} for splay and a bit more recently for Greedy~\cite{fox2011upper}. Note that $k$-increasing sequence and $k$-decreasing sequence are symmetric. 

\begin{theorem}
\label{thm:main k-dec} 
For $k$-increasing or $k$-decreasing input $X$, Greedy serves input $X$ with cost at most $O(\min\{n k^{2}, nk\alpha(n)\})$. 
\end{theorem}

Previously, the best analysis of Greedy achieves the upper bound of $O(n k^2)$~\cite{chalermsook2015pattern}. They showed that the greedy matrix  avoids a permutation of size $k^2$. Furthermore, the permutation is \textit{layered} (i.e., a concatenation of decreasing sequences into layers such that each entry of a layer is smaller than the following layers) and thus it admits $O(nk^2)$ bounds  by Theorem 1.6 of \cite{cibulka2013extremal}.  Our new result improves the previous bound whenever $k > \alpha(n)$. 
\Cref{table:results} summarizes our main results. 

\paragraph{A New Result in Extremal Combinatorics:}  Along the way of proving $k$-increasing bounds for Greedy, we discover a new result in extremal combinatorics regarding the bounds of ``easy'' permutation patterns. A seminal result by Marcus and Tardos~\cite{marcus2004excluded}, show that $\textsf{ex}(n,P) = O(n 2^{k \log k})$ for any length-$k$ permutation matrix $P$. The bounds have been improved to $\textsf{ex}(n,P) = n 2^{O(k)}$ by \cite{fox2013stanley,CibulkaK17}. Furthermore, Fox \cite{fox2013stanley} showed (via a randomized construction) that almost all permutation matrices have the bound $\geq n2^{\Omega((k/\log k)^{1/2})}$  and left open the following conjecture: 

\begin{conjecture}
If $\pi$ is a permutation that avoids $O(1)$-sized pattern, ${\sf ex}(n, \pi) = n \cdot poly(|\pi|)$. 
\end{conjecture}

Here, we make a partial progress by showing an approach to determine if the extremal bound of a permutation matrix $P$ has polynomial dependence on $k$ instead of exponential dependence. As a result, we discover a new class  permutation matrices whose extremal bounds are polynomially bounded. 

For any permutation  $\pi$, denote by ${\sf dleft}(\pi)$ (abbreviation for ``delete from the left'') the permutation  obtained by removing the point (in the matrix form of $\pi$) on the leftmost column as well as its corresponding row and column; for instance, ${\sf dleft}(1,3,4,2) = {\sf dleft}(2,3,4,1) = (2,3,1)$. Similarly, we can define ${\sf dright}(\pi)$. We say that a length-$k$ permutation matrix $P$ is \textit{left-reducible} if it contains a point on one of the two corners of the first column (i.e., at coordinate $(1,1)$ or $(1,k)$). Similarly, we say that $P$ is \textit{right-reducible} if it contains a point on one of the two corner of the last columns (i.e., at coordinate $(k,1)$ or $(k,k)$). 

\begin{definition} 
Let $P$ be a length-$k$ permutation matrix. We say that $P$ is \textit{reducible} to a permutation matrix $Q$, denoted by $P \rightarrow Q$, if one of the followings is true
\begin{itemize}
    \item $Q = {\sf dleft}(P)$ and $P$ is left-reducible, or
    \item $Q = {\sf dright}(P)$ and $P$ is right-reducible. 
\end{itemize} 
Furthermore, we say that $P$ is reducible to $Q$ in $t$ steps, denoted by $P \overset{t}{\rightarrow} Q$, if there exist permutation matrices $P_1, \ldots, P_{t-1}$ such that $P \rightarrow P_1 \rightarrow \ldots \rightarrow P_{t-1} \rightarrow Q$. 
\end{definition}

We say that a length-$k$ permutation $P$  is \textit{$k$-linear} if $\textsf{ex}(n,m,P) = O(k(m+n))$ where $\textsf{ex}(n,m,P)$ is asymmetric extremal bounds of  $m$-by-$n$ matrix avoiding $P$. Similarly, we say that $P$ is \textit{$k$-polynomial} if   $\textsf{ex}(n,P) \leq nk^{O(1)}.$ Our new result is the following. 

\begin{theorem} \label{thm:permutation reduction}
If a length-$k$ permutation $P$ is reducible to a $k$-linear permutation $Q$  in $t$ steps, then  $\textsf{ex}(n,P) \leq nk^{O(t)}.$  In particular, $P$ is $k$-polynomial whenever $t = O(1)$.
\end{theorem}

In other words, if we start with a $k$-linear permutation, we can add a point on one of the corners and repeat for a few times, then the resulting permutation is $k$-polynomial. An example of linear permutation includes an identity matrix (Theorem 7(i) of \cite{BrualdiC21}). Another important class of linear permutation is $\textit{layered permutation}$. A layered permutation is a concatenation of decreasing sequences $S_1, \ldots, S_\ell$ such that every element of $S_i$ is smaller than all elements of $S_{i+1}$. Layered permutations are known to be linear in $k$ (Theorem 1.6 of \cite{cibulka2013extremal}).  We refer to \cite{CibulkaK17} for more discussion regarding $k$-linear permutations.




\vspace{0.2in} 

\noindent 
{\bf Further Related Work.} 
The ``parameterized'' pattern avoiding inputs, where one considers an input class whose avoided pattern has size depending on parameter $k$, have recently received attention (see e.g.,
~\cite{chalermsook2015pattern,goyal2019better}). 
Research questions in this setting aim to first prove the upper bound for $\opt(X)$ as a function of $k$ and later show that Greedy matches this upper bound. 
The parameters of interest are those that generalize the classical special cases of the dynamic optimality conjecture (such as deque and pre-order traversals). 
Chalermsook et al.~\cite{chalermsook2015pattern} showed $O(n 2^{\alpha(n)^{O(|\pi|)}})$ upper bound for the cost of Greedy on inputs avoiding $\pi$. 
Goyal and Gupta proved $O(n \log k)$ upper bound on the cost of Greedy on $k$-decomposable sequences~\cite{goyal2019better} (if one allows ``preprocessing''). A stronger bound that subsumes the $k$-decomposable bound and the dynamic finger bound~\cite{cole2000dynamic1,cole2000dynamic2} was shown by~\cite{iacono2016weighted,bose2016power} (see discussion in~\cite{chalermsook2016landscape}.). 
Besides pattern avoidance, other BST bounds include the unified bound~\cite{buadoiu2007unified,iacono2001alternatives,derryberry2009skip} and the multi-finger bound
~\cite{chalermsook2018multi,howat2013fresh,demaine2013combining}.  
The original drawback of Greedy was that, in contrast to Splay whose simplicity made it attractive for practitioners, Greedy does not admit a simple implementation in the BST model. However, due to a recent work by  Kozma and Saranurak
~\cite{kozma2019smooth}, there exists a heap data structure (called smooth heap) which matches the cost of Greedy and is implementable in practice.  

\paragraph{Conclusion and Open Problems.} 
We propose a simple idea of partitioning the execution log of Greedy into several simpler-to-analyze matrices based on the input structures and leveraging distinct patterns to upper bound each structured matrix separately.  
Based on this idea, we derived improved bounds for many notorious pattern avoidance conjectures (and completely settling the postorder conjecture).
We view these results as a showcase of the decomposition trick, which allows us to extend/strengthen the applications of forbidden submatrices in amortized analysis of data structures. 
We believe that this technique would find further uses in BSTs and more broadly in amortized analysis of data structures. 

\paragraph{Paper Organization.}
We start preliminaries including notations, definitions (the formal definition of Greedy BSTs in particular), basic facts about Greedy in \Cref{section: short-prelims}. In \Cref{section: seq}, we start with a warm-up section providing simple proofs of sequential access theorem, and delete-only deque sequence. In \Cref{sec:size3input}, we prove Greedy bounds for preorder and postorder traversals. In \Cref{sec:deque}, we discuss Deque sequence with insertions and deletions. We discuss Split conjecture for Greedy in \Cref{section: Split}. In \Cref{section: k-dec}, we discuss $k$-decreasing sequence. In \Cref{sec: step2}, we discuss the new result in extremal combinatorics.

\section{Preliminaries} \label{section: short-prelims}


\noindent 
{\bf Matrix and geometry:} Let $M$ be a binary matrix. For geometric reasons, we write matrix entries column-first, and the rows are ordered bottom-to-top, i.e. the first row is the bottom-most. 
Strictly speaking, $M(i,j)$ is the value of $i$-th column and $j$-th row of $M$. The matrix $M$ can be interchangeably viewed as the set of points $\pset(M)$ such that $(i,j) \in \pset(M)$ if and only if $M(i,j) = 1$. 
We abuse the notation and sometimes write $M$ (viewing $M$ as both the matrix and the point set corresponding to $1$-entries) instead of $\pset(M)$. Denote by $|M|$  the number of 1s in $M$.  

For point $p \in [n] \times [m]$, we use $p.x$ and $p.y$ to refer to the $x$ and $y$-coordinates of $p$ respectively. 
Let $I \subseteq [n]$ and $J \subseteq [m]$ sets of consecutive integers. We refer to $R= I \times J$ as a rectangle. We say that matrix $M$ is \textit{empty} in rectangle $R$ if $M(i,j) = 0$ for all $(i,j) \in R$; or equivalently, the rectangle $R$ is \textit{$M$-empty}. 

Let $\sigma = (\sigma(1), \sigma(2), \ldots, \sigma(k))$ be a permutation. We can view any permutation $\sigma$ as a matrix $M_{\sigma}$ where $M_{\sigma}(\sigma(i), i) =1$ and other entries are zero.

\noindent
{\bf Pattern avoidance:} We say that matrix $M$ contains pattern $P$ if $P$ can be obtained by removing rows, columns and non-zero entries of $M$. If $M$ does not contain $P$, we say that $M$ avoids $P$. 
The theory of forbidden submatrices focuses on understanding the following extremal bound: ${\sf ex}(n,P)$ is defined as the maximum number of $1$-entries in an $n$-by-$n$ matrix that avoids $P$. 

We will use the following known bounds from \cite{FurediH92,DBLP:journals/jct/Pettie11}: 

\begin{theorem} [\cite{FurediH92,DBLP:journals/jct/Pettie11}]
\label{thm:exbound}
\begin{itemize}
\item ${\sf ex}\bigg(n, \kbordermatrix{
     &  &  &  \\
      &  & 1 &   \\ 
    &  1  & &  1   }\bigg) = O(n)$. 

\item  ${\sf ex}\bigg(n, \kbordermatrix{
     &  &   & &  \\
      &  & 1 &   &  1 \\ 
    &  1  & &  1  &  }\bigg) = O(n\alpha(n))$.
    
\item     ${\sf ex} \bigg(n, \kbordermatrix{
     &  &  & & \\
      &  & & &1  \\ 
      &  & 1&  &\\ 
    &  1 & & 1&}\bigg) = O(n\alpha(n))$.
    
\item  ${\sf ex}\bigg(n, \kbordermatrix{
     &  &   & & &  \\
      &  & 1 &   &  1 & \\ 
    &  1  & &  1  &  & 1 }\bigg) = O(n2^{\alpha(n)})$.
\end{itemize}
\end{theorem}

\vspace{0.2in}

\noindent 
{\bf Greedy algorithm:}  We consider input in the matrix $X$, that is, $X(i,j) = 1$ if and only if key $i$ is accessed at time $j$. Notice that each matrix row contains exactly one $1$-entry.  
Denote by $Y= G_T(X)$ the matrix corresponding to the execution log of Greedy on sequence $X$ and initial tree $T$, that is, $Y(i,j) = 1$ if and only if key $i$ is touched by Greedy on initial tree $T$ at time $j$.
We have $X \subseteq G_T(X)$ (Greedy always touches the input).  For any two points $p,q \in {\mathbb R}^2$, we denote $\Box_{p,q}$ as the minimally closed rectangular area defined by $p$ and $q$.  

We explain how $G_T(X)$ is constructed. 
Inputs are  matrix $X$ (one point per row) and matrix $T$ ($n$ rows and $n$ columns). The columns of $G_T(X)$ are $[n]$ and the rows of are indexed by $\{-(n-1),\ldots, 0, 1, \ldots, m\}$. The non-positive rows are exactly by $T$.   Greedy starts adding points into $G_T(X)$ by processing rows $t=1,\ldots, m$ in this order. 
At time $t$, we initialize $S \leftarrow \emptyset$. For any key $a \in [n]$, we denote $\tau(a,t)$ as the last time $t'$ before $t$ such that the point $(a,t')$ was added by Greedy or by the initial tree. Let $p$ be an point in $X$ on $t$-th row. For each $a \in [n]$, let $q = (a,\tau(a,t))$. If the  rectangle $\Box_{p,q}$ contains only two points $p$ and $q$, then we add point $(a,t)$ to $S$. 
After we process all keys $a$, we add points in $S$ to $G_T(X)$. 
See Figure~\ref{fig:greedyexample} below. 

We say that points in $X$ are accessed and points in $G_T(X)$ are touched. Moreover, for point $p$ in $G_T(X)$ we say that key $p.x$ is touched at time $p.y$. Throughout the paper, our statements hold for every initial tree $T$, hence we use $G(X)$ instead of $G_T(X)$.

\begin{figure}[ht]
  \centering
  \includegraphics[scale=0.4]{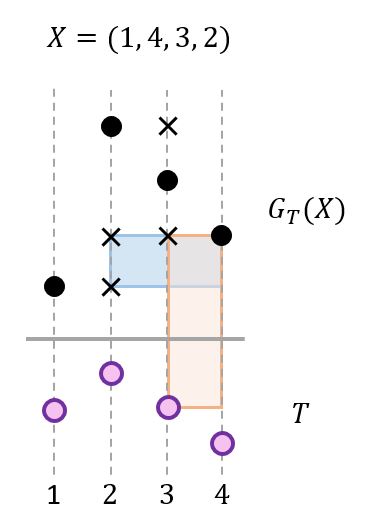}
  \caption{An example of $G_T(X)$ \label{fig:greedyexample}}
\end{figure}

We extend the pattern avoidance notation to handle multiple types of points. Let ``$\bullet$'' denote each input point in $X$ and ``$\times$'' each point in $G(X) \setminus X$. 
The notation $Y=G(X)$ contains $\kbordermatrix{
     &  a &  b  \\
    t_2 & \bullet &   \\ 
    t_1&   & \times }$ is used in the most intuitive way: Entry  $Y(a,t_2)$ contains an input, and $Y(b,t_1)$ contains a touched point (or equivalently, $b$ is touched at time $t_1$). 
    
We are interested in studying the pattern avoidance bound for Greedy. Define $\gex(n,P)$ as the maximum execution cost of Greedy $G(X)$ over all permutation input $X$ that avoids pattern $P$ and over all initial trees $T$. This extremal function is a function of $n$ and $|P|$.

\noindent
{\bf Multi-typed pattern avoidance:} For convenience, we extend the pattern avoidance terminology to allow points to have different types. 
Let $M$ be a point set and ${\mathbb T}$ be the set of type of points. A type function of $M$ is a mapping $f: M \rightarrow {\mathbb T}$. In the matrix view, the type $f$ assigns a value in ${\mathbb T}$ to each non-zero entry of $M$. 

Let $M$ be a matrix and $P$ a pattern. Let $\mu$ and $\pi$ be types of $M$ and $P$ respectively.  
We say that $(M,\mu)$ contains $(P,\pi)$ if and only if
\begin{itemize}
    \item $M$ contains $P$ or there exists a submatrix $M'$ obtained by removing rows, columns, and points of $M$ such that $M' = P$. Also, $\pi': M' \rightarrow {\mathbb T}$ be the type function $\pi$ induced on $M'$. 

    \item For all $(i,j) \in P$, we have $\pi'(i,j) = \mu(i,j)$. 
\end{itemize}

In this paper, our types are ${\mathbb T} = \{\times, \bullet\}$ where $\times$ and $\bullet$ are the touched (but non-accessed) and accessed points respectively.
A Greedy matrix $G(X)$ is naturally associated with a type function $f:G(X) \rightarrow {\mathbb T}$ which assigns $\times$ to points in $G(X) \setminus X$ and $\bullet$ to points in $X$.  
Therefore, in the statement  
\[ (G(X), f) \mbox{ avoids }  \kbordermatrix{
     &  &   \\
      &  & \times  \\ 
    &  \bullet &     }\]  
we will often omit the types and simply say $G(X)$ instead of $(G(X),f)$.

If $M$ contains a $k$-by-$q$ pattern $P$, then there exist columns $c_1 \leq \ldots\leq c_q$ and rows $r_1\leq \ldots \leq r_k$ such that the induced submatrix of $M$ on those contains $P$.  
In such case, we use the following notation to specify such rows and columns where the pattern appears: 
\[M \mbox{ contains }  \kbordermatrix{
     & c_1 &  \ldots & c_q  \\
    r_k  &  &   &  \\ 
     \vdots &   & P(i,j)  &    \\ 
      r_1 &   &  &     }\]

\vspace{0.2in} 

\noindent 
{\bf Input-revealing properties of Greedy:}    We use a small matrix gadget that allows us to ``reveal'' the location of an input point in $X$. 

\begin{claim} [Generic Capture Gadget~\cite{chalermsook2015pattern}]\label{claim: inputRevealing} 
If $G(X)$ contains $\kbordermatrix{
    & a & b & c \\
   t_2  &  &\times  & \\ 
   t_1 &  \times &  & \times  }$,  
then the input matrix $X$ is non-empty in the rectangle $[a+1,c-1] \times [t_1+1,t_2]$. 
\end{claim}

\begin{claim} [One-sided Capture Gadget]\label{hidden} 
If $G(X)$ contains $\kbordermatrix{
     & a & b  \\
     t_2 &  & \times  \\ 
    t_1&  \times &     }$ or $\kbordermatrix{
     & a & b  \\
     t_2 &  & \times  \\ 
    t_1&  \bullet &     }$,
then  input matrix $X$ is non-empty in  rectangle $[a+1,\infty) \times [t_1+1,t_2]$. This  holds symmetrically for  $\kbordermatrix{
     & a & b  \\
     t_2 & \times &   \\ 
    t_1&  & \times    }$ and $\kbordermatrix{
     & a & b  \\
     t_2 & \times &   \\ 
    t_1&   & \bullet     }$.
\end{claim}
\begin{proof}
Let $(c,t')$ be a touched point in the rectangle $[a+1,\infty) \times [t_1+1,t_2]$ with smallest $t'$. We will show that $X$ is non-empty in the rectangle $[a+1,\infty) \times [t_1+1,t']$, which will imply the Claim. Assume the rectangle $[a+1,\infty) \times [t_1+1,t']$ is input-empty. Let $p$ be an input point at time $t'$ such that $p.x \le a$. Since $c$ is touched at time $t'$, the rectangle $[p.x,c]\times[\tau(c,t'),p.y]$ must be empty. This contradicts to the fact that $(a,t_1)\in [p.x,c]\times[\tau(c,t'),p.y]$.
\end{proof}

\begin{figure}[ht]
  \centering
  \includegraphics[scale=0.35]{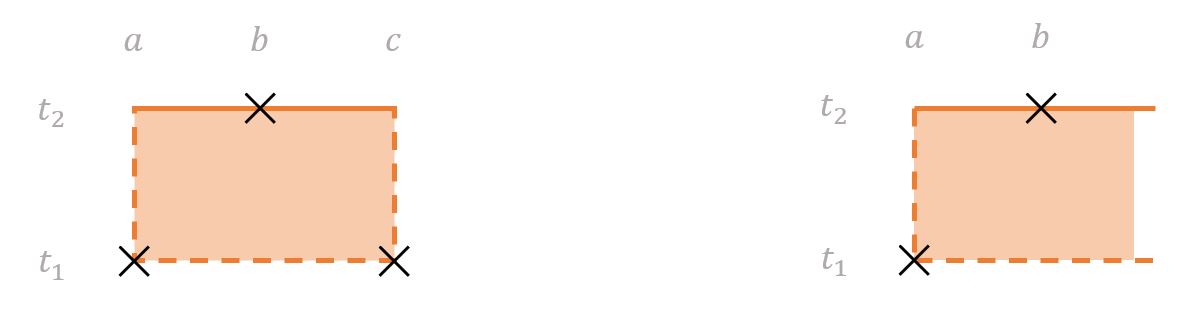}
  \caption{Illustrations of \Cref{claim: inputRevealing} (left) and \Cref{hidden} (right)}
\end{figure}

\begin{corollary} \label{refinedHidden}
If $G(X)$ contains $\kbordermatrix{
     & a & b  \\
     t_2 &  & \times  \\ 
    t_1&  \times &     }$ or $\kbordermatrix{
     & a & b  \\
     t_2 &  & \times  \\ 
    t_1&  \bullet &     }$,
then  input matrix $X$ is non-empty in  rectangle 
$\left( [b,\infty)\times [t_1+1,t_2]\right)$ or in  $([a+1,\infty)\times [t_2,t_2])$. This  holds symmetrically for  $\kbordermatrix{
     & a & b  \\
     t_2 & \times &   \\ 
    t_1&  & \times    }$ and $\kbordermatrix{
     & a & b  \\
     t_2 & \times &   \\ 
    t_1&   & \bullet     }$.
\end{corollary}
\begin{proof}
Assume the input matrix $X$ is empty in the rectangle 
$\left( [b,\infty)\times [t_1+1,t_2]\right)$ and in the rectangle  $([a+1,\infty)\times [t_2,t_2])$. Let $(c,t')$ be the top most input point in $[a+1,b-1] \times [t_1+1,t_2-1]$. From \Cref{hidden}, $X$ must be non-empty in $[a+1,b-1] \times [t_1+1,t_2-1]$. This contradicts to the fact that $(c,t')$ be the top most input point.
\end{proof}

\begin{claim}[Monotone Capture Gadget~\cite{chalermsook2015pattern}]\label{claim:monotoneGadget}
If $X$ avoids $(1,2,...,k)$ and $G(X)$ contains $\kbordermatrix{
    & a_1 &a_2& \cdots & a_{k+1} \\
   t_{k+1}  & \times& &  & \\ 
   \vdots  & &\cdots&&\\
t_2      & &&\times&\\
   t_1 &&   &  &\times}$, then the input matrix $X$ is non-empty in the rectangle $[a_1,a_{k+1}-1] \times [t_1+1,t_{k+1}]$.
\end{claim}

\section{Warm-Up} \label{section: seq}

We present two warm-up proofs before proceeding to our main results: (i) a very short proof for the sequential access theorem of Fox~\cite{fox2011upper}, and (ii) a proof when a given input sequence avoids both $(2,3,1)$ and $(2,1,3)$. 
This is a special case of both preorder traversal and deque conjectures.  


\vspace{0.2in} 

\noindent 
{\bf Sequential Access Theorem:}
Let $X$ be a sequence that avoids $(2,1)$ (equivalently, $X$ is the permutation $(1,2,\ldots, n)$) and $G(X)$ be the Greedy points on $X$. Notice that the points in $X$ lie on the diagonal line $x=y$. 
Decompose $G(X)$ into $X\cup Y_L \cup Y_R$ where $Y_L= \{q \in G(X)\mathrel{}\mid\mathrel{} q.y>q.x\}$ and $Y_R= \{q \in G(X)\mathrel{}\mid\mathrel{} q.y<q.x\}$. In words, the sets $Y_L$ and $Y_R$ are the points strictly on the left and right of the diagonal line  respectively.
\begin{observation} \label{outside}
$Y_L$ avoids $\kbordermatrix{
     &  &   \\
     & \times &   \\ 
    &   & \times }$, so $|Y_L| \leq O(n)$. 
\end{observation}
\begin{proof}
Assume otherwise that $Y_L$ contains $\kbordermatrix{
     & a & b  \\
     t_2 & \times &   \\ 
    t_1&   & \times }$ for some time indices $t_1 < t_2$ and keys $a <b$. Let $p = (b,t_0)$ be the input point below $(b,t_1)$. Applying \Cref{refinedHidden}, there exists an input point $q$ in the region $(-\infty, b-1] \times [t_1+1,t_2]$. Notice that $p$ and $q$ form $(2,1)$ a contradiction.  
\end{proof}

\begin{claim} \label{claimseq}
    $Y_R$ avoids $\kbordermatrix{
     & & &  \\
     & \times & &\times  \\ 
    &   & \times& }$. Therefore, $|Y_R|= O(n)$. 
\end{claim}
\begin{proof}
Assume otherwise that $Y_R$ contains $\kbordermatrix{
     & a & b &c \\
     t_2 & \times & &\times  \\ 
    t_1&   & \times& }$ for some time indices $t_1 < t_2$ and keys $a<b<c$. Let $q$ denote an input point at time $t_2$. Because $(a,t_2) \in Y_R$, we have that $q.x<a$. 
    Applying \Cref{hidden} and the fact that $Y_R$ contains  $\kbordermatrix{
     & b & c  \\
     t_2 &  & \times  \\ 
    t_1&  \times & 
    }$, we would have that the rectangular region $[b+1,\infty) \times [t_1+1,t_2]$ must contain input point $p$, which cannot be the same point as $q$; so we have $p.y < q.y$. The points $p$ and $q$ induce pattern $(2,1)$, a contradiction. See Figure~\ref{fig:seq} for illustration. 
\end{proof}


\begin{figure}[ht]
  \centering
  \includegraphics[scale=0.4]{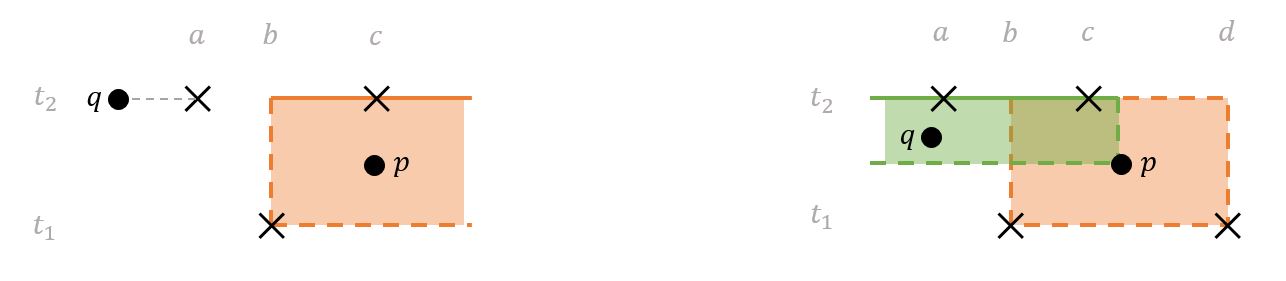}
  \caption{Illustrations of the proofs for \Cref{claimseq} (left) and \Cref{claimpath} (right) \label{fig:seq}} 
\end{figure}

\vspace{0.2in} 

\noindent 
{\bf ``Deque Access'' Theorem:} 
Let $X$ be an input permutation that avoids $(2,3,1)$ and $(2,1,3)$. 
This is a special case of the deque conjecture, roughly equivalent to the case when we are only allowed to delete the minimum and maximum. 
In fact, as we argue later, this is a special case of preorder  traversal, deque, and split  conjectures. 
In this section, we show that $|G(X)| = O(n \alpha(n))$.

Let $r \in X$ be an input point on the top row of $X$. We decompose $G(X)$ into $X \cup G_{<}\cup G_{>} \cup G_{=}$ where 
\begin{itemize}
    \item $G_{<}=\{q \in G(X) \setminus X\mathrel{}\mid\mathrel{} q.x<r.x\}$
    \item $G_{>}=\{q \in G(X) \setminus X\mathrel{}\mid\mathrel{} q.x>r.x\}$
    \item $G_{=}=\{q \in G(X) \setminus X\mathrel{}\mid\mathrel{} q.x=r.x\}$
\end{itemize}

To highlight the structure of this sequence, we similarly break $X$ into $X= \{r\} \cup X_{<} \cup X_{>}$ where $X_< = \{p \in X: p.x < r.x\}$. 

\begin{observation}
The points in $X_{<}$ and $X_{>}$ avoid $(2,1)$ and $(1,2)$ respectively.
\end{observation}

This means that the input points in $X_{<}$ form an increasing sequence, while those in $X_{>}$ form a decreasing sequence.  
Intuitively, we view input points as a triangle without base, then we partition the plane into the vertical column, left side and right side of the triangle. 
\begin{observation}
$|G_{=}| \leq n$.
\end{observation}
Next we will show that $|G_{<}| \le O(n\alpha(n))$. The left-right symmetric arguments also hold for upper bounding $|G_{>}|$. We further decompose $G_{<}$ into $Y_L \cup Y_R$ where $Y_L=\{q \in G_{<} \mathrel{}\mid\mathrel{} \exists s\in X \text{ with } (q.y = s.y) \land (q.x<s.x<r.x)\}$ and $Y_R=G_{<} \setminus Y_L$. In words, the sets $Y_L$ and $Y_R$ are the points that are ``outside" and ``inside" the triangle respectively. 
\begin{observation}
$Y_L$ avoids $\kbordermatrix{
     &  &   \\
     & \times &   \\ 
    &   & \times }$, so $|Y_L| \leq O(n)$. 
\end{observation}
\begin{claim} \label{claimpath}
    $Y_R$ avoids $\kbordermatrix{
     & & & & \\
     & \times & &\times&  \\ 
    &   & \times& &\times}$. Therefore, $|Y_R|= O(n\alpha(n))$. 
\end{claim}
\begin{proof}
Assume otherwise that $Y_R$ contains $\kbordermatrix{
     & a & b &c &d\\
     t_2 & \times & &\times & \\ 
    t_1&   & \times& & \times}$ for some time indices $t_1 < t_2$ and keys $a<b<c<d$. Applying \Cref{claim: inputRevealing} and the fact that $Y_R$ contains  $\kbordermatrix{
     & b & c &d \\
     t_2 &  & \times&  \\ 
    t_1&  \times & &\times
    }$, we would have that the rectangular region $[b+1,d-1] \times [t_1+1,t_2]$ must contain input point $p \in X_{<}$. Because $(a,t_2) \in Y_R$, we have that $p.y<t_2$. Applying \Cref{hidden} and the fact that $Y_R$ contains $\kbordermatrix{
     & a &p.x \\
     t_2  & \times&  \\ 
    p.y  & &\bullet }$, we can conclude that the rectangle $(-\infty,p.x-1] \times [p.y+1,t_2]$ contains an input point $q \in X_{<}$. 
    The points $p$ and $q$ induce pattern $(2,1)$, a contradiction. 
\end{proof} 

\section{Bounds for Input Avoiding Size-$3$ Patterns} \label{sec:size3input}

There are six patterns of size three. We divide them into three different groups as follows: $\Pi_1 = \{(1,2,3), (3,2,1)\}$, $\Pi_2 = \{(2,3,1), (2,1,3)\}$ and $\Pi_3 = \{(1,3,2), (3,1,2)\}$. We  argue that, for each such pattern class $\Pi_i$ (for $i =1,2,3$), we only need to analyze the cost of Greedy on one pattern in $\Pi_i$. 
We make this claim precise as follows. For each matrix (point set) $M$ with $n$ columns, define the flipped matrix $M^{flip}$ as a matrix $M'$ obtained by reflecting $M$ around $y$-axis, that is, $M'(i,j) = M(n-i+1, j)$ for all $i,j$.
Therefore, if we define $P_1 = (1,2,3), P_2 = (2,3,1)$ and $P_3 = (1,3,2)$, we would have $\Pi_i = \{P_i, P_i^{flip}\}$.

\begin{proposition}
For each $i =1,2,3$, we have $\gex(n,P_i) = \gex(n,P_i^{flip})$. 
\end{proposition}

\begin{proof}
We first prove that $\gex(n,P_i) \leq \gex(n,P_i^{flip})$. Let $X$ and $T$ be the input an initial tree that achieves the value $\gex(n,P_i)$. 
First, notice that if $X$ avoids $P_i$, then $X^{flip}$ avoids $P_i^{flip}$. 
The following claim can be proved using the fact that Greedy is symmetric: 
\begin{observation}
For any $X$ and initial tree $T$, we have $(G_T(X))^{flip} = G_{T^{flip}}(X^{flip})$. 
\end{observation}
This observation can be proved, for instance, by induction on the number of rows. 
Therefore, $X^{flip}$ and $T^{flip}$ are the inputs that prove that $\gex(n,P_i^{flip}) \geq \gex(n,P_i)$. The other direction can be argued symmetrically. 
\end{proof}
We will use our techniques to prove the following theorems, which are the restatement of \Cref{thm:main} (1 and 2).

\begin{theorem} [Preorder Traversal] \label{thm:preorder} 

     For $P \in \Pi_2$, $\gex(n,P) = \gex(n,(2,3,1)) = O(n2^{\alpha(n)})$. 
    
\end{theorem}

\begin{theorem} [Postorder Traversal] \label{thm:postorder} 
     For $P \in \Pi_3$, $\gex(n,P) = \gex(n,(1,3,2)) = O(n)$.
\end{theorem}

    
    

As for the input in $\Pi_1$, a theorem of \cite{chalermsook2015pattern} implies that Greedy costs at most $O(n)$. In the following subsections we prove the above theorems.

\subsection{Preorder Traversals} \label{section: 213}

This section is devoted to proving \Cref{thm:preorder}. Let $X$ be an input matrix which corresponds to a permutation that avoids $(2,3,1)$. 
We partition the Greedy points $Y$ into four parts based on the location of the points with respect to the input points. 

\begin{observation}
For each point $q \in Y \setminus X$, there are (unique) input points $p_1, p_2 \in X$ such that $p_1.x = q.x$ and $p_2.y = q.y$. 
\end{observation}

Using this observation, for each such point $q$, if $p_1.y < q.y$, we say that $q$ is a bottom point; otherwise, we say that it is a top point, so this will partition $Y \setminus X$ into $T \cup B$ where $T$ and $B$ are the top and bottom points respectively.  
Similarly, we define the left/right partition $Y \setminus X = L \cup R$ where $L$ and $R$ are the left and right points.  

These can be used to define our partition as follows: $BR = B \cap R, BL = B \cap L, TR = T \cap R, TL = T \cap L$.

\begin{lemma}
\label{lem: traversal BR} 
    BR avoids $\kbordermatrix{
     &  &   \\
     & \times &   \\ 
    &   & \times }$. Therefore, $|BR| = O(n)$. 
\end{lemma}
\begin{proof}
Assume otherwise that BR contains $\kbordermatrix{
     & a & b  \\
     t_2 & \times &   \\ 
    t_1&   & \times }$ for some time indices $t_1 < t_2$ and keys $a <b$.
Since the point $(b,t_1)$ is a touched point in BR, there are input points $p$ and $q$ that are at the bottom and right of $(b,t_1)$ respectively. From Claim~\ref{hidden}, the region $(-\infty,b-1] \times [t_1+1, t_2]$ must contain an input point $r$. The points $p,q$ and $r$ induce pattern $(2,3,1)$ in $X$, a contradiction.
\end{proof}

\begin{lemma}
\label{lem: traversal BL} 
    BL avoids $\kbordermatrix{
     & & &  \\
     & \times & &\times  \\ 
    &   & \times& }$. Therefore, $|BL| = O(n)$. 
\end{lemma}
\begin{proof}
Assume otherwise that  BL contains $\kbordermatrix{
     & a & b &c \\
     t_2 & \times & &\times  \\ 
    t_1&   & \times& }$ for some time indices $t_1 < t_2$ and keys $a < b < c$.
Let $p$ be an input point at the bottom of $(b,t_1)$. Let $r$ be an input on the left of $(a,t_2)$. 
Applying  \Cref{hidden} and the fact that $BL$ contains  $\kbordermatrix{
     & b & c  \\
     t_2 &  & \times  \\ 
    t_1&  \times & 
    }$, we would have that the rectangular region $[b+1,\infty) \times [t_1+1,t_2]$ must contain input point $q$. Notice that $q \neq r$, so we have that $p,q$ and $r$ induce pattern $(2,3,1)$ in $X$, a contradiction.
\end{proof}

\begin{figure}[ht]
  \centering
  \includegraphics[scale=0.4]{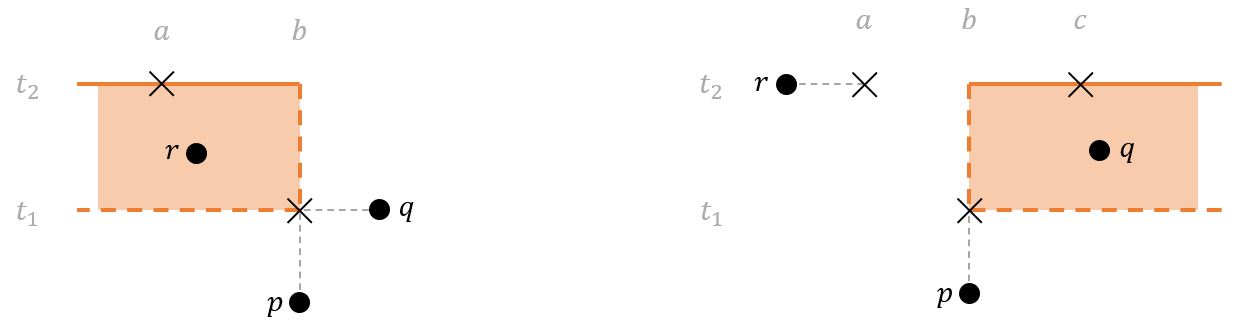}
  \caption{Illustrations of the proofs for $BR$ (left) and $BL$ (right) in preorder\label{fig:prebottom}}
\end{figure}

\begin{lemma} \label{lemma: topavoid}
    $TL \cup TR$ avoids $\kbordermatrix{
     & &  & \\
      & \times & &\times  \\ 
    &   & \bullet& }$. 
\end{lemma}
\begin{proof}
Assume otherwise that $TL \cup TR$ contains $\kbordermatrix{
     & a & b &c \\
     t_2 & \times & &\times  \\ 
    t_1&   & \bullet& }$. 
Let $r$ be an input point at the top of $(a,t_2)$. 
Applying \Cref{hidden} and the fact that $TL \cup TR$ contains $\kbordermatrix{
     & b &c \\
     t_2  & &\times  \\ 
    t_1  & \bullet& }$, we can conclude that the rectangle $[b+1,\infty) \times [t_1+1,t_2]$ contains an input point $q$. Let $p = (b,t_1)$. Notice that $p, q$ and $r$ induce pattern $(2,3,1)$, a contradiction. See Figure~\ref{fig:pretop}.    
\end{proof}

\begin{figure}[ht]
  \centering
  \includegraphics[scale=0.4]{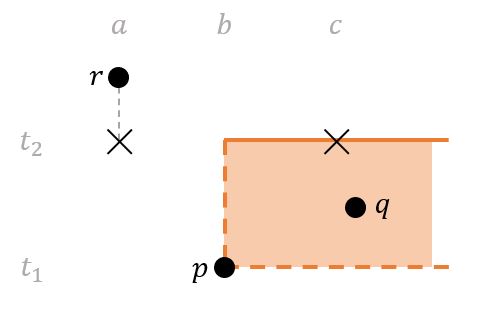}
  \caption{An illustration of the proof for $TL \cup TR$ in preorder. \label{fig:pretop}}
\end{figure}

\begin{corollary} \label{claim: TR}
    Each TL and TR avoids $\kbordermatrix{
     &  &  & & &  \\
      & \times & &\times && \times \\ 
    &   & \times& &\times&}$. Therefore, $|TL| + |TR| \leq O(n2^{\alpha(n)})$. 
\end{corollary}
\begin{proof}
We only present the proof for $TL$; the arguments for $TR$ are symmetric. Assume otherwise that TL contains $\kbordermatrix{
     & a & b &c&d&e \\
     t_2 & \times & &\times && \times \\ 
    t_1&   & \times& &\times&}$.
Applying Claim~\ref{claim: inputRevealing} to the pattern $\kbordermatrix{
      & b &c&d \\
     t_2  & &\times &  \\ 
    t_1   & \times& &\times }$, there must be an input point $q \in [b+1,d-1] \times [t_1+1, t_2]$. If $q.y < t_2$, we would be done since it would contradict Lemma~\ref{lemma: topavoid}, so assume that $q.y = t_2$. Since $a < q.x < e$, this implies that $(a,t_2) \in R$, a contradiction.    
\end{proof}

\subsection{Postorder Traversals} \label{section: 132}

This section is devoted to proving \Cref{thm:postorder}. Let $X$ be a permutation that avoids $(1,3,2)$. 
We partition the Greedy points $G(X)$ into four sets $BL, BR,TL, TR$ in the same way as in the last subsection. 

\begin{lemma}
BR avoids $\kbordermatrix{
     &  &  &  \\
      &  & \times&  \\ 
    &  \times & & \times}$. Therefore, $|BR| \leq O(n)$. 
\end{lemma}
\begin{proof}
Assume otherwise that BR contains $\kbordermatrix{
     & a & b &c \\
     t_2 &  & \times&  \\ 
    t_1&  \times & & \times}$ for some time indices $t_1 < t_2$ and keys $a < b < c$. 
Let $p= (a,t_0)$ be an input at the bottom of $(a,t_1)$. Let $q=(c',t_1)$ be an input on the right of $(c,t_1)$. 
Using the Claim~\ref{claim: inputRevealing}, there must be an input point $r$ in the rectangle $[a+1,c-1] \times [t_1+1, t_2]$.  The points $p,q$ and $r$ induce pattern $(1,3,2)$ in $X$, a contradiction. 
\end{proof}

\begin{lemma}
BL avoids $\kbordermatrix{
     & &  & \\
      &  & \times&  \\ 
    &  \times & & \times}$. Therefore, $|BL| \leq O(n)$. 
\end{lemma}
\begin{proof}
Assume for contradiction that BL contains $\kbordermatrix{
     & a & b &c \\
     t_2 &  & \times&  \\ 
    t_1&  \times & & \times}$ for some time indices $t_1 < t_2$ and keys $a< b< c$.
Let $p=(a,t_0)$ be the input point at the bottom of $(a,t_1)$. Applying \Cref{refinedHidden} with the submatrix  $\kbordermatrix{
     & a & c  \\
     t_1 &   & \times \\ 
    t_0&  \bullet  & }$, it follows that there exists input point $q$ in the rectangle $[c,\infty) \times [t_0+1,t_1]$ or $([a+1,\infty)\times [t_1,t_1])$. Since an input point at time $t_1$ has to be on the left of $a$, This means $q$ must be in the rectangle $[c,\infty) \times [t_0+1,t_1-1]$. Finally, applying Claim~\ref{claim: inputRevealing}, there exists input point $r$ in the rectangle $[a+1,c-1] \times [t_1+1,t_2]$. The points $p,q$ and $r$ induce $(1,3,2)$ in $X$, a contradiction.
\end{proof}

\begin{figure}[ht]
  \centering
  \includegraphics[scale=0.4]{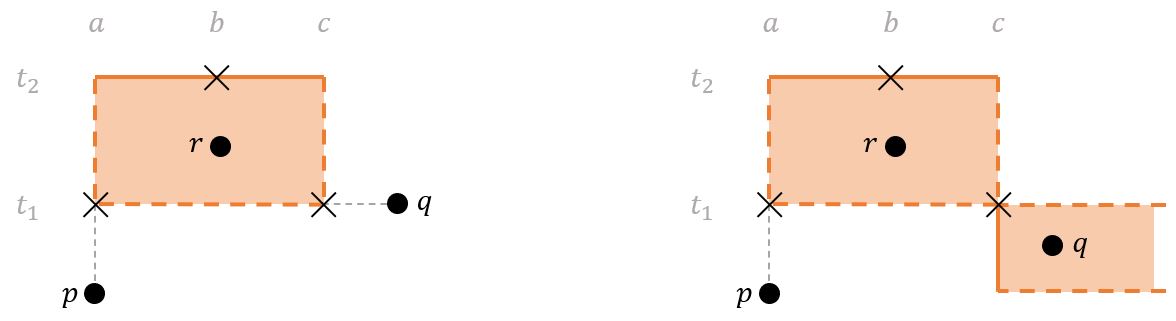}
  \caption{Illustrations of the proofs for $BR$ (left) and $BL$ (right) in postorder}
\end{figure}

\begin{lemma} \label{lem: post TR}
    TR avoids $\kbordermatrix{
     &  &  & \\
      & \times & &\times  \\ 
    &   & \times& }$. Therefore, $|TR| \leq O(n)$.
\end{lemma}
\begin{proof}
Assume otherwise that TR contains $\kbordermatrix{
     & a & b &c \\
     t_2 & \times & &\times  \\ 
    t_1&   & \times& }$ for some time indices $t_1 < t_2$ and keys $a<b<c$. 
Since $(c,t_2)$ is in TR, there are input points $q$ and $r$ at the right and top of it respectively.
Applying \Cref{hidden} to the submatrix $\kbordermatrix{
     & a & b  \\
     t_2 & \times &   \\ 
    t_1&   & \times }$, there must be an input point $p$ in the region $(-\infty, b] \times [t_1+1,t_2]$. 
Since $X$ is a permutation, we have that $p.y < q.y$ (in particular, $p \neq q$). The points $p,q$ and $r$ induce $(1,3,2)$ in $X$, a contradiction. 
 
\end{proof}

\begin{lemma} \label{lem: post TL}
TL avoids $\kbordermatrix{
     &  &   \\
     & \times & \times  \\ 
    & \times  &  }$. Therefore, $|TL| \leq O(n)$.
\end{lemma}
\begin{proof}
Assume otherwise that TL contains $\kbordermatrix{
     & a & b  \\
     t_2 & \times & \times  \\ 
    t_1& \times  &  }$ for time indices $t_1 < t_2$ and keys $a<b$. 
Let $r$ be an input at the top of $(a,t_2)$, and $p$ be the input at the left of $(a,t_1)$. Using \Cref{hidden} for the submatrix $\kbordermatrix{
     & a & b  \\
     t_2 & & \times  \\ 
    t_1& \times  &  }$, there must exist point $q$ in the region $[b,\infty) \times [t_1+1,t_2]$.  
The points $p,q$ and $r$ induce $(1,3,2)$ in $X$, a contradiction.
\end{proof}

\begin{figure}[ht]
  \centering
  \includegraphics[scale=0.4]{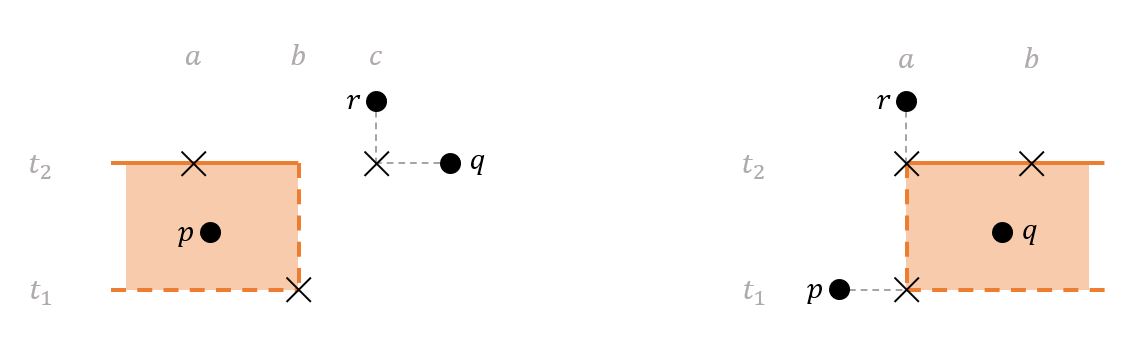}
  \caption{Illustrations of the proofs for $TR$ (left) and $TL$ (right) in postorder}
\end{figure}

\section{Dynamic Deque with Insertion and Deletion}\label{sec:deque}

In this section we prove \Cref{thm:main}(\ref{item:main1}). We use the model for dynamic update deque sequence with insertion and deletion from \cite{DBLP:journals/corr/ChalermsookG0MS15a}. An update sequence $S$ is a set of points where each point is either inserted, deleted or accessed. For our purposes, our model only deals with insertions and deletions and no access points.

\begin{definition}[Deque Sequence] An update sequence is a deque sequence if it only consists of {\sc InsertMin, InsertMax, Deletemin, DeleteMax} operations.
\end{definition}

In this model, each key can be inserted and followed by deletion at most once. In addition, it can be touched only during the time between insertion and deletion. More precisely, for any key $p$, let $t_{ins}(p)$ and $t_{del}(p)$ be an insertion and deletion time of $p$, respectively.  If $p$ is in an initial tree $T$, $t_{ins}(p)=0$. The model ensures that $t_{ins}(p) < t_{del}(p)$. The active time $act(p)$ is the interval of time $[t_{ins}(p), t_{del}(p)]$, and $p$ can be touched during $act(p)$.

Let $Min_t$ and $Max_t$ be the set of keys which are deleted by {\sc DeleteMin} and {\sc DeleteMax} before time $t$, respectively.

\begin{definition}[from \cite{DBLP:journals/corr/ChalermsookG0MS15a} Concentrated Deque Sequence] A deque sequence is concentrated if, for any time $t$, the inserted element $x$ is the minimum, then $y<x$ for all $y \in Min_t$, and if $x$ is the maximum, then $x<y$ for all $y \in Max_t$.
\end{definition}

\begin{figure}[ht]
  \centering
  \includegraphics[scale=0.3]{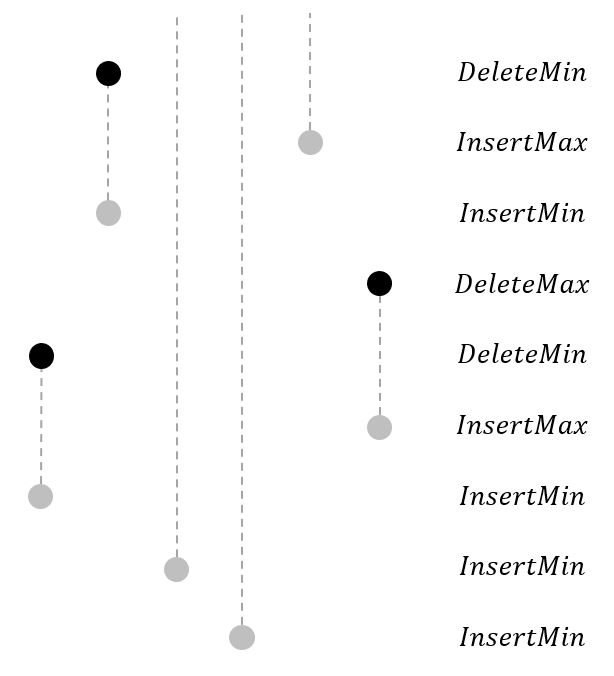}
  \caption{An example of concentrated deque sequence}
\end{figure}

In the succeeding, we will use the following lemma from \cite{DBLP:journals/corr/ChalermsookG0MS15a} to assume that the updated deque sequences we work with are always concentrated deque sequences.

\begin{lemma} [From \cite{DBLP:journals/corr/ChalermsookG0MS15a}]For any deque sequence $S$ , there is a concentrated deque sequence $S'$ such that the execution of any BST algorithm on $S'$ and $S$ have the same cost.
\end{lemma}

\subsection{The $O(m \alpha (n))$ Bound}
Let $X$ with $|X| = m$ be an input instance of concentrated deque with insertion and deletion where $x_i$ is the inserted or deleted key at time $i$.
A key $a$ is touched (excluding deletions) by Greedy at a time $t$ only when $\square_{(a, \tau(a,t)), (x_t, t)}$ is empty.
\begin{lemma}
\label{lemma:insert-delete_deque}
 Assuming $m \ge n$, then $|G(X)| \leq O(m \alpha (n))$.
\end{lemma}

Let us assume that the initial tree $T$ has $k_0$ active keys. We divide the execution of Greedy into phases. The first phase starts at time $t=1$ and lasts till time $t=\frac{k_0}{2}$. After the end of the first phase, let $k_1$ be the number of active keys. Then, our second phase starts at time $t=\frac{k_0}{2}+1$ and lasts till time $t=\frac{k_0}{2}+\frac{k_1}{2}$.  This process continues from one phase to another until no active keys remain or $x_m$ is inserted or deleted. When there are no active keys at the beginning of a phase, we wait for the first key to be inserted and then begin our phase. Also, we always consider $\lceil \frac{k}{2} \rceil$ for $k$ active keys in any phase but for notation we will use $\frac{k}{2}$.

For the first phase, we see Greedy's execution on $X$ till time $\frac{k_0}{2}$. At time $\frac{k_0}{2}$, we divide the touched point of Greedy into two parts \emph{left} (denoted as $\lset_{\frac{k_0}{2}}$) and \emph{right} (denoted as $\rset_{\frac{k_0}{2}}$) such that both parts contains equal number of active keys. $\lset_{\frac{k_0}{2}}$ contains all keys to the left of the divide and  $\rset_{\frac{k_0}{2}}$ contains all the keys to the right of the divide. We show that the number of touched point in $\lset_{\frac{k_0}{2}}$ and $\rset_{\frac{k_0}{2}}$ is $O(k_0\alpha(n))$.

In general, if there are $k_i$ active keys at the start of phase $i$ then we show that the number of points touched by Greedy in phase $i$ is $O(k_i\alpha(n))$. Summing over all phases, $\sum\limits_{\text{phase $i$}} k_i\alpha(n)=O(m\alpha(n)).$

\subsection{Greedy adds $O(k_i\alpha(n))$ points in phase $i$}

\begin{lemma}\label{lemma:leftpoints}
$\lset_{\frac{k_i}{2}}$ avoids $P = \kbordermatrix{
& a &b  &c   &d \\
t'& \times  &  & \times &  \\ 
    t& &\times  &  & \times }$.
\end{lemma}
 
\begin{proof}
Let us assume for  contradiction that $\lset_{\frac{k_i}{2}}$ contains $P$.
Applying Claim \ref{claim: inputRevealing} and the fact that $P$ contains $\kbordermatrix{
     & b & c &d \\
     t' &  & \times&  \\ 
    t&  \times & &\times
   }$ we would have that the rectangular region $[b+1,d-1] \times [t+1,t']$ must contain an input point $q\in X$. Since there are at most $\frac{k_i}{2}$ {\sc InsertMax} and {\sc DeleteMax} together in phase $i$, point $q$ cannot be an operation {\sc InsertMax} or {\sc DeleteMax}. Next, we will show that $q.y \in act(b)$, which implies that $q$ cannot be an operation {\sc InsertMin} or {\sc DeleteMin} because $q.x>b$. 
   
   To show that $q.y \in act(b)$, it suffices to show that $t_{del}(b) \geq t'$. If $t_{del}(b)$ is not in phase $i$, the statement trivially holds. Assume for contradiction that $t_{del}(b)$ is in phase $i$ and $t_{del}(b) < t'$. This means $b$ gets deleted by operation {\sc DeleteMin}. There are two cases: $t_{del}(b) \in act(a)$ and $t_{del}(b) \notin act(a)$. In the first case, {\sc DeleteMin} cannot delete $b$ since $a$ is active. In the second case, it means that $t_{ins}(a) > t_{del}(b)$, which contradicts to the fact that $X$ is concentrated sequence.
   


\end{proof}

Similar to above lemma we can prove that $\rset_{\frac{k_i}{2}}$ avoids $P' = \kbordermatrix{
& a &b  &c   &d \\
    t'& &\times  &  & \times \\
    t& \times  &  & \times & }$. This implies that the number of points in $\lset_{\frac{k_i}{2}}$ and $\rset_{\frac{k_i}{2}}$ is $O(\frac{k_i}{2} \alpha(\frac{k_i}{2}))$ (using Theorem \ref{thm:exbound}). Thus, the total number of points added by Greedy in phase $i$ is $O(k_i \alpha(\frac{k_i}{2}))$
As $n\ge k_i$, this quantity is $O(k_i \alpha(n))$.

\section{Split Model} \label{section: Split}

In this section we prove \Cref{thm:main split}. The geometric view of Greedy is invented for the purpose of search
~\cite{demaine2009geometry} and insert/delete~\cite{DBLP:journals/corr/ChalermsookG0MS15a}. We will first define Greedy execution in the split model and show the relation between this model and the standard search model when $X$ avoids some patterns. 
Our main theorem in this section is the following:  

\begin{theorem}
\label{thm: split main} 
Let $X \in [n]^n$ be a permutation. Then, 
\begin{itemize}
    \item If $X$ avoids $(1,3,2)$ and $(2,3,1)$, the cost of Greedy's deleting $X$ is at most the cost of Greedy's spliting $X$. 
    \item For any sequence $X$, there exists a sequence $X'$ avoiding $(2,3,1)$ such that Greedy's spliting $X$ costs at most Greedy's searching $X'$. 
\end{itemize}
\end{theorem}

\begin{corollary}
For any permutation $X$, Greedy's splitting $X$ costs at most $O(n 2^{\alpha(n)})$. 
\end{corollary}

\begin{corollary}
For Greedy, the traversal conjecture implies the split conjecture, which implies the deque conjecture.  
\end{corollary}

\subsection{The Split Model} 

\newcommand{\iset}{{\mathcal I}} 
\newcommand{\tset}{{\mathcal T}} 

Let $X = (x_1, x_2, … , x_n) \in [n]^n$ be a permutation input sequence of keys where $x_i$ is split at time $i$. 
Let $\iset_X = \{I_{x_1}, I_{x_2}, … , I_{x_n}\}$ be the set of intervals defined  as follows. 
First, we create a binary search tree $\tset_X$ by inserting the keys of $X$ into an empty initial tree where $x_i$ is inserted at time $i$.\footnote{When inserting $x_i$, we search the current tree $\tset_X$ until a miss occurs, and we insert $x_i$ at the corresponding place.}  
Define $I_{x_i}$ as the minimal open integer interval containing all keys in the subtree rooted at $x_i$ in $\tset_X$. Notice that $I_{x_1} = (0,n+1)$. 




See \Cref{fig:preorder} for illustration. These intervals define the ``active keys'' for each key, that is, $I_a$ is the interval containing active keys when key $a$ is split. 

\begin{figure}[h] 
  \centering
  \includegraphics[scale=0.3]{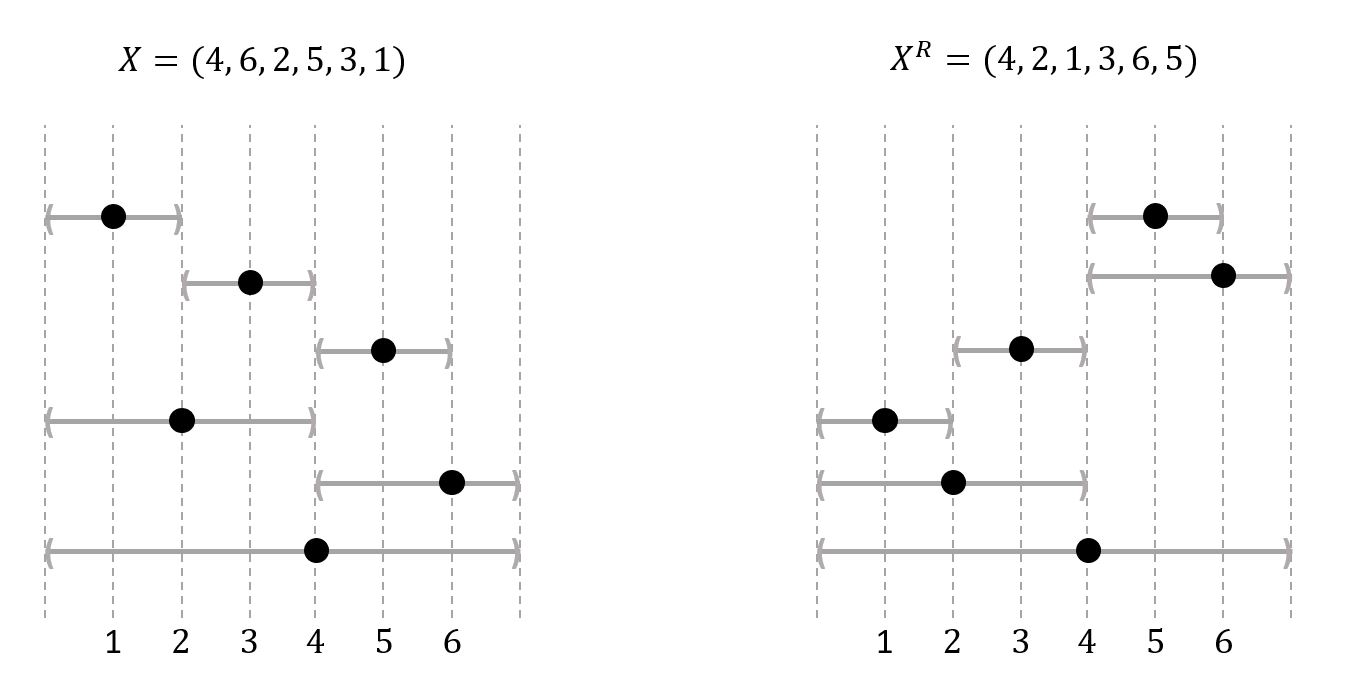}
  \caption{An example of $X$ and $X^R$ obtained from Algorithm~\ref{alg:rearrangement} \label{fig:preorder}}
\end{figure}

\begin{observation} \label{containitself}
For each $i \in [n]$, $x_i \in I_{x_i}$.
\end{observation}
\begin{observation} \label{notcontainprevious}
For any $i < j$, $x_i \notin I_{x_j}$.
\end{observation}
\begin{observation} \label{laminar}
$\iset_X$ is a laminar family of intervals, i.e., two intervals intersect if an only if one is completely contained in another. Furthermore, for $j>i$, it is either $(I_{x_j} \subset I_{x_i})$ or $(I_{x_j} \cap I_{x_i} = \emptyset)$.
\end{observation}

Let $\tau(a,t)$ denote the last touched time of key $a$ before time $t$. When $t$ is clear from the context, we use $\tau(a)$.  Let $\square_{(a_1,b_1), (a_2,b_2)}$ denote the \textit{closed} rectangular area defined by two points: $(a_1,b_1)$ and $(a_2,b_2)$. We define the Greedy execution on input $X$ in the split model, $G_S(X)$, as in Algorithm~\ref{alg:split}:

\begin{algorithm}[hbt!]
\caption{Greedy execution in the split model $G_S(X)$}\label{alg:split}
  \DontPrintSemicolon
Given $X$ and $\iset_X$\;

\For{$i\gets1$ \KwTo $|X|$}{
$S=\{a \in I_{x_i}\mathrel{}\mid\mathrel{}\square_{(a, \tau(a)), (x_i, i)} \text{ is empty}\}$\;
        
    $\forall a \in S,$ add point $(a, i)$ to $G_S(X)$\; 
}
\end{algorithm}

\subsection{Relation to Preorder Traversals} 

The second part of~\Cref{thm: split main} follows from the following lemmas. 

\begin{lemma}
\label{lem: split-traversal rearrange} 
Given a permutation input sequence $X$, there exists a preorder permutation input sequence $X^R$ such that $|G_S(X^R)|=|G_S(X)|$.
\end{lemma}

\begin{lemma}
\label{lem: split at most search} 
Let $X$ be a preorder sequence. Then, $|G_S(X)| \leq |G(X)|$.
\end{lemma}

\subsubsection*{Proof of Lemma~\ref{lem: split-traversal rearrange}} 

For a permutation input instance $X = (x_1, x_2, …, x_n)$, we denote by $X^R = (x^R_1, x^R_2, …, x^R_n)$ its \textit{rearranged permutation input instance} which we construct in algorithm \ref{alg:rearrangement} (\Cref{fig:preorder}). 

Let $B= (b_1,\ldots, b_k)$. 
Denote by {\sc Swap}($B,i$) the operation that swaps $b_i$ with $b_{i+1}$, that is, it returns $B'$ which is the same as $B$ everywhere except for $b'_i = b_{i+1}$ and $b'_{i+1} = b_i$. Then $X^R$ is obtained by iteratively applying {\sc Swap}. 
We argue below that $X^R$ is a preorder traversal of binary search tree $\tset_X$. 

\begin{algorithm}[hbt!]
\caption{Rearrange $X$ into $X^R$}\label{alg:rearrangement}

  \DontPrintSemicolon
  \SetKwFunction{FMain}{Main}
  \SetKwFunction{FPre}{Preorder}
  
    \SetKwProg{Pn}{Function}{:}{\KwRet}
  \Pn{\FPre{$B=(b_1,...,b_k)$}}{
        \While{$\exists i: (b_{i+1} < b_{i}) \land (I_{b_i} \cap I_{b_{i+1}}=\emptyset)$}{
            $B' \gets \text{\sc Swap}(B,i)$ 
        }
        \KwRet $B'$\;
  }
  \;
  \SetKwProg{Fn}{Function}{:}{}
  \Fn{\FMain{$X$}}{
        $X^R \gets  \FPre{X}$ \;
        \KwRet $X^R$\;
  }

\end{algorithm}



\begin{lemma}
\label{lem: invariant under swap} 
Let $B= (b_1,\ldots,b_k)$. If $I_{b_i} \cap I_{b_{i+1}}  = \emptyset$, and $B'$ is obtained by {\sc Swap}($B,i$). Then $\tset_{B} = \tset_{B'}$ and hence $\iset_B = \iset_{B'}$. 
\end{lemma}

In other words, this lemma proves that the BST is invariant under the swap operation. 

\begin{proof}
Since the intervals $I_{b_i}$ and $I_{b_{i+1}}$ are disjoint, let $c$ be the LCA of $\tset_B$ at the moment before time $i$ (i.e. before inserting $b_i$). 
Notice that $b_{i+1}$ is inserted into the right subtree of $c$, while $b_i$ is inserted into the left subtree of $c$, and the order of their insertions do not matter. Therefore, $\tset_B$ and $\tset_{B'}$ would be the same after time $i+1$. 
\end{proof}

\begin{claim}
$X^R$ is preorder sequence. In particular, it is a preorder traversal of $\tset_X$. 
\end{claim}
\begin{proof}
First, we argue that $X^R$ avoids $(2,3,1)$. Assume otherwise that it contains $i < j <k$ such that $x^R_i, x^R_j, x^R_k$ induce $(2,3,1)$, so we must have that $b_j > b_i > b_k$. Let $j': j \leq j' < k$ be the minimum integer such that $b_{j'} > b_{j'+1}$ (notice that such $j'$ must exist). Notice that $I_{j'}$ ends before $b_i$ while $I_{j'+1}$ starts after $b_i$ so they are disjoint. 
This implies that the swap would have been applied at $j'$, a contradiction. 
Since $X^R$ is a preorder permutation, it must be a preorder permutation of $\tset_{X^R}$. 
From~\Cref{lem: invariant under swap}, we have $\tset_{X} = \tset_{X^R}$. 
\end{proof}

\begin{lemma}
$|G_S(B')| = |G_S(B)|$.
\end{lemma}
\begin{proof}
Let $M[t]$ denote the row $t$ of matrix $M$ (recall that this paper start indexing from the bottom most row). It is easy to see that $G_S(B')[t] = G_S(B)[t]$ when $t<i$ because both sequences are similar up to time $i-1$. Next, we claim that $G_S(B')[i] = G_S(B)[i+1]$. This is because $I_{b'_{i}} = I_{b_{i+1}}$ and $G_S(B)(p,i)=0$ for all $p \in I_{b_{i}}$. So, for any key $a \in [n]$, $\square_{(a, \tau(a)), (b'_i, i)}$ is empty in $G_S(B')$ if and only if $\square_{(a, \tau(a)), (b_{i+1}, i+1)}$ is empty in $G_S(B)$. Similar argument holds for $G_S(B')[i+1] = G_S(B)[i]$.

Lastly, we claim that $G_S(B')[t] = G_S(B)[t]$ when $t>i+1$. Let $r$ be the first time after $i+1$ such that $G_S(B')[r] \ne G_S(B)[r]$. We claim that, for any key $a \in [n]$, $\square_{(a, \tau(a)), (b'_r, r)}$ is empty in $G_S(B')$ if and only if $\square_{(a, \tau(a)), (b_{r}, r)}$ is empty in $G_S(B)$. There are 4 cases:
\begin{enumerate}
    \item if $\tau(a)>i+1$ in $G_S(B)$, this is trivial by our assumption.
    
    \item if $\tau(a)=i+1$ in $G_S(B)$, $\square_{(a, i+1), (b_{r}, r)}$ is empty in $G_S(B)$ if and only if $\square_{(a, i), (b'_r, r)}$ is empty in $G_S(B')$. This is because $G_S(B)(p,i)=0$ for all $p \in I_{b'_{i+1}}$.
    
    \item when $\tau(a)=i$ in $G_S(B)$, this is symmetric to the above case.
    
    \item when $\tau(a)<i$ in $G_S(B)$. Notice that the only difference between $G_S(B)$ and $G_S(B')$ before time $r$ are in rows $i$ and $i+1$. One can view $\square_{(a, \tau(a)), (b_r, r)}$ as a set of consecutive columns. Since $G_S(B')[i] = G_S(B)[i+1]$ and $G_S(B')[i+1] = G_S(B)[i]$ , this means $\square_{(a, \tau(a)), (b'_r, r)}$ is empty in $G_S(B')$ if and only if $\square_{(a, \tau(a)), (b_{r}, r)}$ is empty in $G_S(B)$.
\end{enumerate}



\end{proof}

\subsubsection*{Proof of~\Cref{lem: split at most search}} 


\begin{figure}[ht]
  \centering
  \includegraphics[scale=0.40]{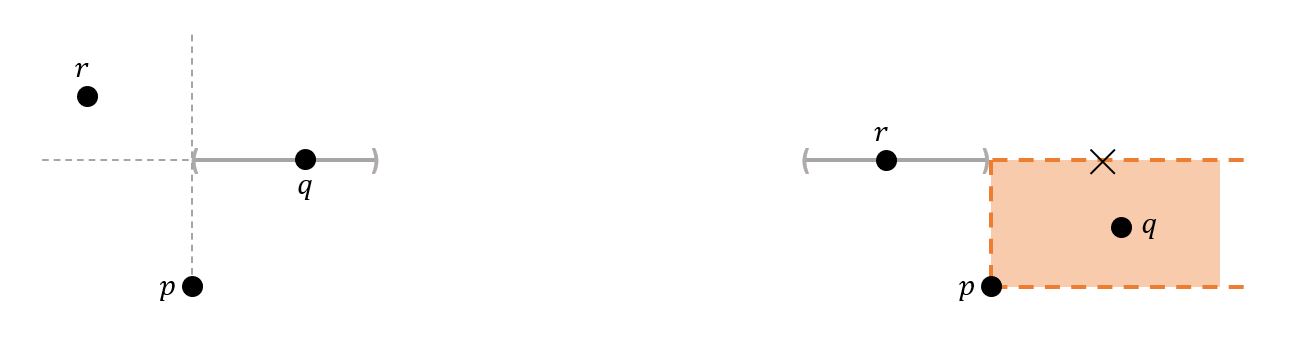}
  \caption{Illustrations of the proofs of \Cref{northwest} (left) and \Cref{east} (right)}
\end{figure}

\begin{lemma}\label{northwest}
Let $X$ be a preorder sequence. For each $q \in X$, there is no point $r \in X$ such that $r.x<\lleft(I_q)$ and $r.y>q.y$. 
\end{lemma}
\begin{proof}
If $\lleft(I_q) \leq 1$, the lemma trivially holds. Consider the case where $\lleft(I_q) > 1$. Assuming such $r$ exists. By interval construction, there is an input point $p$ such that $p.x = \lleft(I_q)$ and $p.y < q.y$. This means points $p,q$ and $r$ form $(2,3,1)$ in $X$. Contradiction.
\end{proof}

\begin{lemma}\label{east}
Let $X$ be a preorder sequence. For each $r \in X$, there is no point $c \in G(X)$ such that $c.x>\rright(I_r)$ and $c.y = r.y$.
\end{lemma}
\begin{proof}
If $\rright(I_r) \geq n$, the lemma trivially holds. Consider the case where $\rright(I_r) < n$. By interval construction, there is an input point $p$ such that $p.x = \rright(I_r)$ and $p.y < r.y$. If $p.y=r.y-1$, there is no such touch point $c$ because $\square_{(c.x, \tau(c.x)), (r.x, r.y)}$ must contain $p$. If $p.y<r.y-1$, assuming there is such touch point $c$. Using \Cref{hidden} with $p$ and $c$, we have that the rectangle $[p.x+1, \infty)\times[p.y+1,r.y-1]$ must contain some input $q$ ($q.y \ne r.y$ because $X$ is permutation). This means points $p,q$ and $r$ form $(2,3,1)$ in $X$. Contradiction.
\end{proof}

For $i \in [n]$, let $G^{(i)}(X)$ and $G^{(i)}_S(X)$ denote a set of points in row $i$ of $G(X)$ and $G_S(X)$, respectively.

\begin{lemma}
Let $X$ be a preorder sequence. For $i \in [n], G^{(i)}(X) \cap I_{x_i} = G^{(i)}_S(X) \cap I_{x_i}$.
\end{lemma}
\begin{proof}
Let $j$ be the first time that $G^{(j)}(X) \cap I_{x_j} \ne G^{(j)}_S(X) \cap I_{x_j}$. This means there exists a point $c \in G(X)\setminus G_S(X)$ such that $c.x \in I_{x_j}$ and $c.y < j$. Let $x_t$ be an input point at time $c.y$. From \Cref{laminar} and the fact that $c \notin I_{x_t}$, we have that $I_{x_t} \cap I_{x_j} = \emptyset$. We divide into two cases: 1) $\rright(I_{x_t}) \leq \lleft(I_{x_j})$ and 2) $\rright(I_{x_j}) \leq \lleft(I_{x_t})$. In the first case, $x_t$ and $c$ contradict \Cref{east}. In the second case, $x_t$ and $x_j$ contradict \Cref{northwest}.
\end{proof}

\section{$(k-1)$-Decreasing Sequences} \label{section: k-dec}

This section is devoted to proving \Cref{thm:main k-dec}. We focus on $(k-1)$-decreasing sequences; the argument for $(k-1)$-increasing sequences is symmetric. The $O(nk^2)$ bound follows from \cite{chalermsook2015pattern} and Theorem 6.1 of \cite{cibulka2013extremal}.  We focus on proving the new $O(kn\alpha (n))$ bound, which is smaller than $nk^2$ whenever $k > \alpha(n)$.




\subsection{An $O(kn\alpha(n))$ bound} 
\label{sec: step1} 
Let $X$ be $(k-1)$-decreasing sequence, i.e.,  a sequence that avoids $I_k = (1,2, \ldots, k)$. For any two points $p, q$, we say that $p$ \textit{dominates} $q$ (denoted by $p \succ q$) if $p.x > q.x$ and $p.y > q.y$. 
Let $q \in G(X) \setminus X$ be a touched, non-input point. 
We define $\textsf{chain}(q)$ to be zero if $q$ is not dominated by any input points in $X$. 
Otherwise, $\textsf{chain}(q)$ is the maximum length $j$ such that there exists input points $p_1,\ldots, p_j \in X$ such that $p_1 \succ \ldots \succ p_j \succ q$; we call $\{p_1,\ldots, p_i\}$ a witness of the fact that ${\sf chain}(q) \geq i$. 
Since $X$ avoids $(1,\ldots,k)$, we have $\textsf{chain}(q) \leq k-1$ for all $q \in Y \setminus X$. 
For $i \in \{0,1,\ldots, k-1\}$, we define $G_i(X) = \{ q \in G(X) \setminus X  \colon \textsf{chain}(q) = i\}$. 
By definitions, we can partition $G(X)$ into $k$ parts. That is, $G(X) \setminus X  = \bigcup_{0 \le i \le k-1}G_i(X)$, and thus $|G(X)| = |X| + \sum_{0 \leq i \leq k-1}|G_i(X)|$. 
So, it suffices to bound each matrix $G_i(X)$ separately.


\begin{proposition} \label{pro:gi avoids 123}
For all $i$, $G_i(X)$ avoids $\kbordermatrix{
     &  &  & \\
      &  & & \times  \\ 
      &  & \bullet&  \\ 
    &  \times & & }$.
\end{proposition}
\begin{proof}
Suppose that $G_i(X)$ contains $\kbordermatrix{
     &  a & b & c\\
      t_3 &  & & \times  \\ 
      t_2&  & \bullet&  \\ 
    t_1&  \times & & }$  for some keys $a < b <c$ and time indices $t_1 < t_2 < t_3$. Denote $p = (c,t_3), q = (a,t_1)$ and $r = (b,t_2) \in X$. Since $p\in G_i(X)$, $\textsf{chain}(p)= i$, and so there are input points $p_1 \succ \ldots \succ p_i$ that dominate $p$, which means they dominate $r$.  The set $\{p_1,\ldots, p_i, r\}$ is a witness that ${\sf chain}(q) \geq i+1$, a contradiction.  
\end{proof}


\begin{corollary} \label{cor:nalpha}
For all $i$, $G_i(X)$ avoids $\kbordermatrix{
     &  &  & & \\
      &  & & &\times  \\ 
      &  & \times&  &\\ 
    &  \times & & \times&}$. Therefore, $|G_i(X)| = O(n \alpha(n))$. 
\end{corollary}
\begin{proof}
Suppose that $G_i(X)$ contains $\kbordermatrix{
     & a &b  &c &d \\
      t_3&  & & &\times  \\ 
     t_2 &  & \times&  &\\ 
    t_1&  \times & & \times&}$. By Claim~\ref{claim: inputRevealing}, there is an input $(b',t_2') \in X$ in the rectangle  $[a+1,c-1]\times [t_1+1,t_2]$.   Therefore, $G_i(X)$  contains $\kbordermatrix{
     &  a & b' & d\\
      t_3 &  & & \times  \\ 
      t_2'&  & \bullet&  \\ 
    t_1&  \times & & }$, contradicting to Proposition \ref{pro:gi avoids 123}. 
\end{proof}

     

\section{Extremal Combinatorics}
\label{sec: step2} 

This section is devoted to proving \Cref{thm:permutation reduction}. Let $P$ be the length-$k$ permutation in the statement. For the rest of this section, we assume that $n$ is a power of $4k^2$. This will imply the theorem (removing the assumption incurs a multiplicative factor of $O(k^2)$).

\subsection{Marcus-Tardos Recurrence (Rephrased)}

We explain Marcus-Tardos  approach \cite{marcus2004excluded} in our language that would allow us to prove our bounds.  
We first introduce another extremal function $f$ that roughly captures the maximum number of rows in a matrix,  avoiding a specified pattern, that contains sufficiently many number of $1$s per row. 

\begin{definition} \label{def:long vector pattern}
For any permutation  $\pi$ and integer $c$, we define $f(c,\pi)$  to be the maximum number of rows $r$ such that there exists a matrix $M$ with $r$ rows and $c$ columns such that (i) each row has at least $2 |\pi|$ many $1$'s and (ii) $M$ avoids $\pi$. 
\end{definition}

Notice that this definition enforces $c \geq 2 |\pi|$ with a trivial base case: 

\begin{observation}
For any permutation $\pi$, we have $f(2 |\pi|,\pi) = |\pi|-1$. 
\end{observation}

Let $\pi$ be any permutation. Let  $\pi'$ be the permutation matrix after rotating the matrix induced by $\pi$ by 90 degree counterclockwise. Denote $|\pi| = k$.
Marcus and Tardos relate the upper bound on ${\sf ex}(n,\pi)$ to the extremal property of $f$ and prove that $f(k^2,\pi) = O(k {k^2 \choose k})$ for any permutation $\pi$. We rederive the bound in terms of $f$. 

\begin{theorem}
\label{thm:ex and f} 
$\textsf{ex}(n,\pi) = O\left(nk^3(f(4k^2,\pi)+f(4k^2,\pi'))\right)$
\end{theorem}

The rest of this section is devoted to proving \Cref{thm:ex and f}. 

Their result can be rephrased in the following way: 
\begin{lemma} \label{lem:recurrence permutation}
$\textsf{ex}(n, \pi) \leq (2k-1)^2 \textsf{ex}(n/4k^2, \pi) + 4nk^2 \cdot f(4k^2,\pi)  +  4nk^2 \cdot f(4k^2,\pi')$.
\end{lemma}
\begin{proof}
The proof closely follows~\cite{marcus2004excluded}. 
Let $M$ be a matrix that avoids $\pi$ and $|M| = \textsf{ex}(n,\pi)$. Let $n' = n/(4k^2)$. We divide the columns of $M$ into $n'$ groups of consecutive columns of size $4k^2$, and similarly we divide the rows of $M$ into $n'$ groups of  consecutive rows of size $4k^2$. Let $B_{ij}$ be the submatrix of $M$ formed by $i$-th group of columns and $j$-th group of rows respectively. We can view $M$ as a block matrix $(B_{ij})_{i,j \in [n']}$. Each block has size $(4k^2) \times (4k^2)$. We say that a row (or a column) is \textit{empty} if all entries are all zero. A matrix is \textit{empty} if all rows and columns are empty.

For each block $B_{ij}$,  we say that $B_{ij}$ is \textit{wide} if it contains at least $2k$ non-empty columns. Also, we say that $B_{ij}$ is \textit{tall} if it contains  at least $2k$ non-empty rows. Let $M'$ be the $n' \times n'$ matrix where $M'(i,j) = 1$ if and only if $B_{ij}$ is non-empty, not wide and not tall. 
Let $T$ be the $n' \times n'$ binary matrix where $T(i,j) = 1$ if and only if $B_{ij}$ is tall. Let $W$ be the $n' \times n'$ binary matrix where $W(i,j) = 1$ if and only if $B_{ij}$ is wide. Observe that each non-empty block can be tall or wide or neither, the three matrices $M'$, $T$, and $W$ covers all blocks $B_{ij}$ from $M$. 
More precisely, we have $M'(i,j) = T(i,j) = W(i,j) = 0$ if and only if $B_{ij}$ is empty. Therefore, 
\begin{align}
  |M| \leq (2k-1)^2|M'| + 16k^4|T|+16k^4|W|.  
\end{align}
The coefficient of the term $|M'|$ is $(2k-1)^2$  because the number of ones in a non-wide and non-tall block is at most $(2k-1)^2$.  The coefficient for both $T$ and $W$ is $16k^4$ because every block is $4k^2 \times 4k^2$. It remains to bound the number of 1's in $M',W$ and $T$. 

\paragraph*{The number of 1's in $M'$.}
\begin{claim}
$M'$ avoids $\pi$. Therefore, $|M'| \leq \textsf{ex}(n',\pi)$. 
\end{claim}
\begin{proof}
Suppose $M'$ contains $\pi$.  This means that there is a set of non-empty blocks $\{ B_{ij} \}$ in $M$ such that we can form the pattern $\pi$ by taking one 1's per block in the set. Therefore, $M$ contains $\pi$, a contradiction.
\end{proof}

\paragraph*{The number of 1's in $W$ and in $T$.}
\begin{claim}
$|W| \leq n'f(4k^2, \pi)$, and  $|T| \leq n'f(4k^2,\pi')$.
\end{claim}
\begin{proof}
 Since $W$ has $n'$ columns, it is enough to show that  the number of 1's in each column of $W$ is at most $f(4k^2,\pi)$. 
 We fix an arbitrary $j$-th column of $W$. Let $\ell$ be the number of 1's in the $j$-th column of $W$ (This means there are $\ell$ wide blocks in the $j$-th group of column in $M$.) 
Assume for the sake of contradiction  that $\ell > f(4k^2,\pi)$.

For any matrix $P$, we define the flattening of $P$ (denoted by ${\sf flat}(P)$) as a binary row-vector $v$, where $v(k) = 1$ if and only if $k$-th column of $P$ is non-empty. 
Let $Q_j$ be the binary matrix obtained by flattening of all the blocks in the $j$-th column-group. More formally, matrix $Q_j$ has $(4k^2)$ columns and $n'$ rows where each row $i \in [n']$ is a flattening of $B_{ij}$. 

\begin{observation}
If $Q_j$ contains permutation $\pi$, then $M$ contains $\pi$. 
\end{observation}

Observe that the number of 1's in each row of $Q_j$ is at least $2k$ because $B_{ij}$ is wide.  Since $\ell > f(4k^2,\pi)$, and each row of $Q_j$ has at least $2k$ many 1's,  Definition \ref{def:long vector pattern} implies that $Q_j$ contains $\pi$, which implies that $M$ contains $\pi$, a contradiction.  
\end{proof}
\end{proof}

We are ready to prove \Cref{thm:ex and f}.
\begin{proof} [Proof of \Cref{thm:ex and f}]
Let $T(n) = \textsf{ex}(n,\pi)$, $q = 4k^2$, and $g(q) = f_{2}(q,\pi)+f_{2}(q,\pi')$. By \Cref{lem:recurrence permutation}, we have 
\begin{align*}
    T(n) &\leq (2k-1)^2T(n/q) + nq(g(q)) \\
    & \leq nq(g(q)) + \frac{(2k-1)^2}{q}nq(g(q)) + (\frac{(2k-1)^2}{q})^2 nq(g(q)) + \ldots  \\ 
    & \leq nq(g(q)) (\sum_{i\geq 0}(\frac{(2k-1)^2}{q})^i) \\ 
    &\leq nq(g(q))k.
\end{align*}
The last inequality follows since $(2k-1)^2/q = (2k-1)^2/(4k^2) < 1$. 
\end{proof}

\subsection{An Upper bound for Function $f$} 



For any permutation  $\pi$, denote by ${\sf dleft}(\pi)$ (abbreviation for ``delete from the left'') the permutation  obtained by removing the point (in the matrix form of $\pi$) on the leftmost column as well as its corresponding row and column; for instance, ${\sf dleft}(1,3,4,2) = {\sf dleft}(2,3,4,1) = (2,3,1)$. 
Similarly, we can define ${\sf dright}(\pi)$. 

\begin{theorem}[Reduction rules]
\label{lem: reduction} 
Let $\pi$ be a length-$k$ permutation whose corresponding permutation matrix contains a point on one of the two corners of the first column (i.e. at coordinate $(1,1)$ or $(1,k)$). Then  $f(c,\pi) = O(c) \cdot f(c,\widehat{\pi})$ where $\widehat{\pi}$ is any permutation such that ${\sf dleft}(\pi) = {\sf dleft}(\widehat{\pi})$. Similarly, if $\pi$ contains a point on two corners of the last column (i.e. coordinate $(k,1)$ or $(k,k)$), then   $f(c,\pi) = O(c) \cdot f(c,\widehat{\pi})$ where ${\sf dright}(\pi) = {\sf dright}(\widehat{\pi})$. 
\end{theorem}





The rest of this section is devoted to proving the proof of~\Cref{lem: reduction}. We prove the case when the permutation matrix of $\pi$ contains $(1,1)$; other cases are symmetric.

\begin{lemma}
$f(c, \pi) \leq f(c-1, \pi) + f(c, \widehat{\pi}) + 2$ 
\end{lemma}
\begin{proof}
Let $M$ be the matrix with $r$ rows and $c$ columns that achieve the bound $f(c,\pi)$, and assume for contradiction that $r > f(c-1, \pi) + f(c,\widehat{\pi}) + 2$. 

Let $M_1$ be the submatrix containing the bottom $f(c-1,\pi) +1$ rows of $M$. This implies that there must be a point $p$ in the first column of $M_1$; otherwise, the submatrix of $M_1$ without its first column would give a contradiction: It has $(c-1)$ columns, $\alpha$ points per row and avoids $\pi$.  

Now let $M_2$ be the $f(c,\widehat{\pi})+1$ top rows of $M$. We know that this submatrix must contain submatrix $\widehat{P}$; moreover, the submatrix of $M_2$ without its first column must contain pattern ${\sf dleft}(\widehat{\pi})$; let us say that the points $Q \subseteq M_2$ induce this pattern. 
This implies that $\{p\} \cup Q$ induces pattern $\pi$, a contradiction. 
\end{proof}

Now \Cref{lem: reduction} follows by simply applying the recurrence at most $c$ times to unfold the term $f(c,\pi)$: 
\[f(c,\pi) \leq f(2k, \pi) + c \cdot f(c,\hat{\pi}) + 2c \] 
Since $f(2k, \pi) = k$, we have that $f(c,\pi) \leq k + c f(c,\widehat{\pi})+ 2c \leq 4c \cdot f(c, \widehat{\pi})$ as desired.  

\subsection{Proof of \Cref{thm:permutation reduction}}

 We first state an important lemma that will be useful for proving \Cref{thm:permutation reduction}. Recall that $Q$ (in the statement of \Cref{thm:permutation reduction}) is a $k$-linear permutation. Note that $Q'$  denotes the rotation of $Q$ by 90 degrees. 
\begin{lemma} \label{lem: f Q small}
$f(c,Q)$ and $f(c,Q') \leq c$. 
\end{lemma}
\begin{proof}
 Let $r = f(c,Q)$. Let $M$ be a matrix with $c$ columns and $r$ rows and avoids $Q$.  Observe that $$  2 k r \leq |M| \leq \textsf{ex}(r,c, Q).$$ 
The LHS holds since we have $2k$ ones per row in $M$. The RHS follows because $M$ avoids $Q$. Since $Q$ is $k$-linear, $\textsf{ex}(r,c, Q) \leq k(r+c)$. So, $2 k r  \leq k(r+c)$, and we have $r \leq c$. The proof for $f(c,Q')$ is similar. 
\end{proof}

We are now ready to prove \Cref{thm:permutation reduction}. 
\begin{proof} [Proof of \Cref{thm:permutation reduction}]

 we have the following. 
\begin{align*}
    \textsf{ex}(n,P) &= O(nk^3 (f(4k^2,P)+f(4k^2,P'))) \\ 
     &= O(nk^3(f(4k^2,Q)\cdot (4k^2)^t + f(4k^2,Q')\cdot (4k^2)^t)) \\
     &= O(nk^3 (4k^2)^{t+1})\\
     &= O(n4^{t}k^{2t+5}).
\end{align*}
The first inequality follows from \Cref{thm:ex and f}. The second inequality follows from \Cref{lem: reduction} and the fact that $P$ is reducible to $Q$ in $t$ steps. The third inequality follows from \Cref{lem: f Q small}. 
\end{proof}

\section*{Acknowledgement}
This project has received funding from the European Research Council (ERC) under the European Union's Horizon 2020 research and innovation programme under grant agreement No 759557.  
   
\bibliographystyle{alpha}
\bibliography{references}

\appendix

\section{Counterexamples} \label{sec: counter examples}

\begin{figure}[ht]
  \centering
  \includegraphics[scale=0.6]{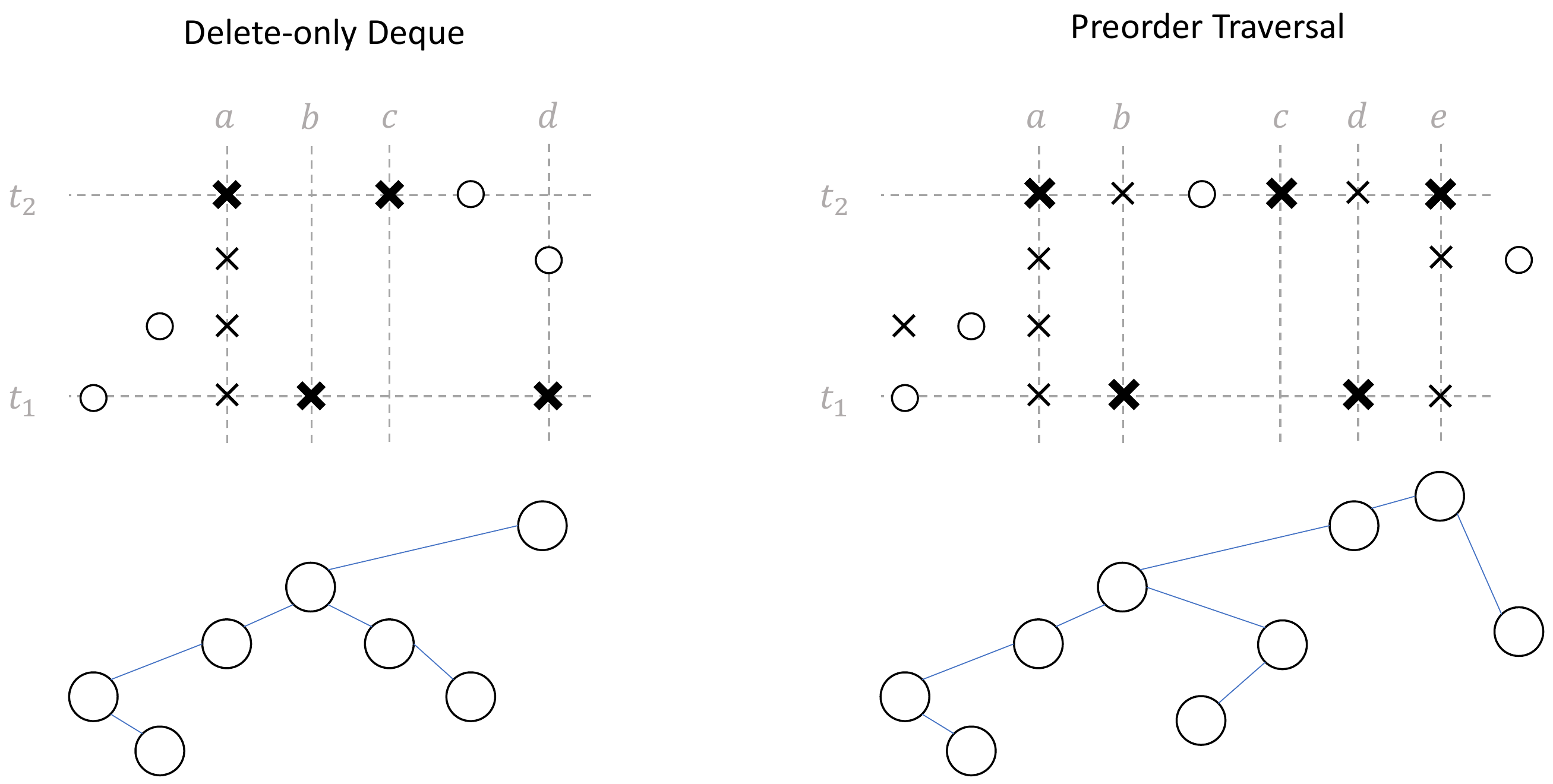}
  \caption{(Left) Greedy on Delete-only Deque Sequences with initial tree below time $t_1$. (Right) Greedy on Preorder Traversal with initial tree below time $t_1$  .}
\label{fig:ce}
\end{figure}

 The counter examples are shown in  \Cref{fig:ce}. The definition of Greedy is discussed in \Cref{section: short-prelims}. In  \Cref{fig:ce} (left), the points $(a,t_2),(b,t_1), (c,t_2)$ and $, (d,t_1)$ form the pattern   $\kbordermatrix{ &  &  & &   \\ & 1 & &1 & \\  &   & 1& & 1}$. In  \Cref{fig:ce} (right), the points $(a,t_2),(b,t_1),(c,t_2),(d,t_1)$ and $(e,t_2)$ form the pattern $\kbordermatrix{ &  &  & & &  \\ & 1 & &1 && 1 \\  &   & 1& & 1 &}$. 

\end{document}